\documentclass[aps,pra,twocolumn,nofootinbib]{revtex4-2}

\let\chap\S

\PassOptionsToPackage{dvipsnames}{xcolor}
\usepackage{xcolor}
\usepackage{algpseudocode}
\usepackage{amsfonts}
\newcounter{algocounter}
\newcommand{\algocaption}[1]{%
  \refstepcounter{algocounter}%
  \par\vspace{3pt}\noindent\textbf{Algorithm \thealgocounter.} #1\par\vspace{3pt}%
}
\usepackage{tikz}
\usepackage{tikz-3dplot}
\usepackage{graphicx}
\usepackage{bm}
\usepackage{physics}
\usepackage{amsthm}
\usepackage{multirow}
\usepackage{amsmath,amssymb}
\usepackage{mathrsfs}
\usepackage{textcomp}
\usepackage{booktabs}
\usepackage{subcaption}
\usepackage{listings}
\usepackage[hidelinks]{hyperref}

\allowdisplaybreaks

\usetikzlibrary{
  positioning, calc, fit, backgrounds,
  arrows.meta, decorations.pathreplacing,
}
\usetikzlibrary{arrows.meta, calc}

\definecolor{OliveGreen}{rgb}{0, 0.6, 0}

\DeclareMathOperator*{\argmax}{arg\,max}

\newcommand{\X}{\mathcal{X}}
\newcommand{\Y}{\mathcal{Y}}
\newcommand{\Z}{\mathcal{Z}}
\newcommand{\Rho}{\mathcal{R}}
\renewcommand{\L}{\mathcal{L}}
\renewcommand{\H}{\mathcal{H}}

\newcommand{\N}{\mathcal{N}}
\newcommand{\F}{\mathcal{F}}
\newcommand{\mS}{\mathcal{S}}

\newcommand{\revise}[1]{{\color{black}#1}}

\theoremstyle{plain}
\newtheorem{theorem}{Theorem}
\newtheorem{proposition}[theorem]{Proposition}
\newtheorem{lemma}[theorem]{Lemma}

\theoremstyle{definition}

\newtheorem{remark}{Remark}
\newtheorem{corollary}{Corollary}
\newtheorem{definition}{Definition}

\begin{document}

\title{An Information-Theoretic Principle for Optimal Quantum Encoding: \\ Tight Frames and Equiangular Ensembles}

\author{Farhad Farokhi}
\email{farhad.farokhi@unimelb.edu.au}
\author{Shuixin Xiao}
\email{shuixin.xiao@unimelb.edu.au}
\affiliation{Department of Electrical and Electronic Engineering, The University of Melbourne, Parkville, VIC 3010, Australia}

\begin{abstract}
Optimal encoding of classical data for quantum-assisted statistical inference
is investigated from an information-theoretic perspective. We prove that the accuracy of any quantum-computing inference procedure is upper bounded by the maximal quantum leakage from the classical data through its quantum encoding, establishing leakage as a universal, task-agnostic quality
measure for encoders. The optimal  encoding strategy, i.e., an encoding strategy that maximizes the maximal quantum leakage, is proved to be attained by pure states.  When there are enough qubits, basis encoding is proved to be universally optimal. However, when the dimension of the system is small, phase encoding is optimal. For the latter, the optimal encoding is not unique; any tight frame, any ensemble whose average state is the maximally mixed state, is in fact optimal. Within tight frames, equiangular tight frames (ETFs) are distinguished as the uniquely symmetric optimal encodings, i.e., they saturate the Welch lower bound on pairwise overlaps. Prominent special cases are the qubit trine, the regular simplex, and symmetric informationally complete positive operator-valued measures (SIC-POVMs), for which the ETF structure and explicit codeword constructions are provided. Numerical examples are presented to validate the theoretical predictions.
\end{abstract}

\maketitle

\section{Introduction}
\label{sec:intro}
Encoding classical data into quantum systems is the first step in virtually every quantum
computing and quantum communication protocol.
In quantum machine learning~\cite{PhysRevA103032430}, raw data must be mapped to quantum states before being fed into parameterized quantum circuits.
In quantum key distribution~\cite{bennett2014quantum}, bits are encoded in non-orthogonal quantum states to guarantee information-theoretic security.
In quantum state transfer over spin systems~\cite{haselgrove2005optimal}, the encodings must be chosen to maximize fidelity at the receiver. In all these settings the choice of encoding is consequential. Not all encodings are equal.

A systematic theory of \emph{optimal} quantum encoding has remained elusive, partly because ``optimal'' depends heavily on the figure of merit. Prior work has studied optimal encoding for fidelity~\cite{haselgrove2005optimal},
for retrieval~\cite{elron2007optimal}, for communication over specific channel
families~\cite{korzekwa2022encoding}, for security under gentle measurements~\cite{farokhi2024measuring},
and for incompatibility-based key sharing~\cite{mitra2021optimal}.
However, a universal, task-independent figure of merit for quantum encoders, i.e., an encoding that simultaneously bounds performance across a wide class of inference problems, has not
previously been identified.

In this paper, we adopt \emph{maximal quantum leakage}~\cite{Farokhi_PRA} as that figure
of merit, and pursue it to a complete characterization of the optimal encoder.
Maximal quantum leakage was introduced in~\cite{Farokhi_PRA} as the largest multiplicative increase in an adversary's guessing probability that can result from any measurement on the quantum encoding of a classical random variable. This notion relates to measured Sibson mutual information of order infinity. It satisfies all the axiomatic requirements of a rigorous information-leakage
measure, i.e., positivity, independence, and the post-processing inequality, and is
independent of both the distribution of $X$ and the specific inference task.

To justify the choice of maximal quantum leakage as a figure of merit for quantum encoding, we prove that the accuracy of any statistical inference problem of any quantum inference procedure is bounded above by the maximal quantum leakage. This bound is in fact tight in the sense that there exists at least one inference problem that saturates the presented upper bound. Maximal quantum leakage depends only on the encoding and not on the inference task, which renders its maximizer the universal optimal encoder for quantum inference algorithms.
We subsequently prove that maximal quantum leakage is equal to the optimal success probability in minimum-error quantum state discrimination with equal priors.
Maximizing leakage is therefore equivalent to designing codewords that are
maximally {distinguishable}. We prove that pure states are optimal for encoding under the developed figure merit and characterize the optimal encoding for a range of parameters. An important observation is that tight frames are optimal when the dimension of the quantum system is small enough. Among all tight frames equiangular tight frames (ETFs) have the smallest overlap with each other and saturate the Welch bound on all pairwise overlaps. ETFs are the most symmetric and most robust optimal encodings. We provide closed-form expressions for a range of parameters. 

The connection between tight frames and quantum state discrimination has appeared in the frame-theory literature~\cite{eldar2002optimal}, but its application to the {design} of quantum codewords for a universal inference criterion is new to this paper.

The rest of the paper is organized as follows. 
Section~\ref{sec:prelim} establishes notation and definitions. Section~\ref{sec:inference} presents the statistical inference model, the
universal accuracy bound, and the relationship with state discrimination. Section~\ref{sec:optimal} presents the optimal encoding. Section~\ref{sec:constructions} establishes the relationship with tight frames and presents explicit constructions for a range of parameters. Section~\ref{sec:numerical} reports numerical examples. Section~\ref{sec:conclusion} concludes the paper.

\section{Preliminaries}
\label{sec:prelim}
\paragraph{Notation.}
All logarithms are in binary basis, motivated by the notion of bits in information theory. Random variables are denoted by capital Roman letters, e.g., $X$. A discrete random variable $X$ with finite alphabet $\X$ is characterized by its
probability mass function $\mathbb{P}\{X=x\}>0$ for $x\in\X$. The restriction to $\mathbb{P}\{X=x\}>0$, $\forall x\in\X$, is without loss of generality as any realization with zero probability can be removed with no impact. We write $|\X|$ to denote the number of distinct elements or the cardinality of the set $\X$.

\paragraph{Quantum states and measurements.}
Let $\H$ denote a finite-dimensional complex Hilbert space of dimension $d:=\dim\H$. The set of linear operators on $\H$ is $\L(\H)$. A \emph{density operator} is a positive semi-definite operator $\rho\in\L(\H)$ with $\trace(\rho)=1$; the set of all density operators is $\mS(\H)$. A density operator $\rho$ is \emph{pure} if $\rank(\rho)=1$, equivalently if $\rho= \ket{\psi} \bra{\psi}$ for some unit vector $\ket{\psi}\in\H$. 

\revise{A positive operator-valued measure (POVM) is a finite collection $\{F_y\}_{y\in\mathcal{Y}}\subset\mathcal{L}(\mathcal{H})$ satisfying $F_y\geq 0$, $\forall y\in\mathcal{Y}$, and $\sum_{y\in\mathcal{Y}}F_y=I$, where $I$ is the identity operator on $\mathcal{H}$. We use $I_d$, instead of $I$, for the identity operator whenever the dimension $d$ needs to be emphasized.} By Born's rule, the probability of outcome $y$ when measuring state $\rho$ is $\mathbb{P}\{Y=y\} =\trace(F_y\rho)$.

\revise{A linear map $\mathcal{N}$ is positive if it maps positive semi-definite operators to positive semi-definite operators, and completely positive if $I_k\otimes\mathcal{N}$ is positive for every integer $k\geq 1$. A quantum channel $\mathcal{N}:\mathcal{S}(\mathcal{H})\to\mathcal{S}(\mathcal{H}')$ is a completely positive and trace-preserving linear map~\cite{wilde2013quantum}.}

\paragraph{Maximal quantum leakage.}
Consider encoding classical data $X\in\X$  into a quantum system $A$  by preparing the system in state $\rho^x\in \mS(\H)$ when $X=x$, $\forall x\in\X$. The collection $\Rho=\{\rho^x\}_{x\in\X}$ is called the \emph{quantum encoding} of $X$.

\begin{definition}[Maximal Quantum Leakage {\cite{Farokhi_PRA}}]
\label{def:qml}
The maximal quantum leakage from $X$ through quantum system $A$ is
\begin{align}\label{eqn:def_qml}
    \mathcal{Q}(X\to A)_\rho
    :=\sup_{\{F_y\}_{y\in\Y}}
    \log\left(\sum_{y\in\Y}
    \max_{x\in\X}\trace(\rho^x F_y)\right),
\end{align}
where the supremum is over all POVMs with arbitrary finite outcome set $\Y$.
\end{definition}

Maximal quantum leakage captures the largest multiplicative increase in the
probability of correctly guessing an arbitrary function of $X$ that any measurement on
the quantum encoding can produce~\cite{Farokhi_PRA}.
It satisfies the post-processing inequality, i.e., $\mathcal{Q}(X\to A)_{\N(\rho)}\leq
\mathcal{Q}(X \to A)_\rho$ 
for any quantum channel $\N$~\cite[Proposition~3]{Farokhi_PRA}. Furthermore, maximal quantum leakage is known to be upper bounded as $\mathcal{Q}(X \to A)_\rho\leq\min\{\log (N),2\log (d)\}$~\cite[Proposition~2]{Farokhi_PRA}.

\paragraph{Minimum-error quantum state discrimination.}
Consider a quantum system prepared in state $\rho_i$ with prior probability $q_i$, $i=1,\dots,N$. A POVM $\{M_i\}_{i=1}^N$ is applied to identify the state with the probability of success given by $\sum_i q_i\trace(M_i\rho_i)$. The optimal success probability is
\begin{subequations}\label{eqn:pguess}
    \begin{align}
    P_{\rm guess}\left( \{q_i,\rho_i\}_{i=1}^N\right)
    \;:=\;
    &\max_{\{M_i\}}\sum_{i=1}^N q_i\trace(M_i\rho_i)\\
    &\;\; \text{s.t.}\,
    M_i\geq 0,\;\sum_{i=1}^N M_i=I.
\end{align}
\end{subequations}
This semi-definite program is efficiently solvable~\cite{barnett2009quantum, bae2015quantum}. 

\paragraph{Frames.}
A collection of vectors $\{\ket{\psi_x}\}_{x\in\X}$ in $\H$ is a \emph{frame} if it spans $\H$, i.e., if the frame operator
$S := \sum_{x\in\X} \ket{\psi_x}\bra{\psi_x}$ is invertible. The frame is \emph{tight} with frame bound $A$ if $S=AI$, and \emph{unit-norm tight} if additionally each $\ket{\psi_x}$ is a unit vector and $A = N/d$. Throughout this paper, ``tight frame'' means unit-norm tight frame unless stated otherwise. More information about frames can be found in~\cite{waldron2018frames}.

\begin{figure*}
    \centering
    \begin{tikzpicture}[scale=.8]
		\node[rectangle, fill=orange!20, minimum width=7.9cm, minimum height=3cm, scale=.8] at (2.35,.3) {}; 
        \node[rectangle, fill=OliveGreen!20, minimum width=4.5cm, minimum height=3cm, scale=.8] at (8.9,.3) {}; 
		\node[rectangle, fill=red!20, minimum width=2.75cm, minimum height=3cm, scale=.8] at (-3.13,.3) {}; 		
		\node[scale=.8] at (6.48,1.4) {+};
		\node[scale=.8] at (2.3,1.4) {Quantum Circuit};
		\node[scale=.8] at (8.9,1.4) {Classical Computing};
		\node[scale=.8] at (-3.1,1.4) {Problem Setup};		
		\node[draw, rectangle, rounded corners=2mm, minimum width=2.5cm, minimum height=2cm, scale=.8] at (-3.15,0) {
			\begin{minipage}{2.3cm}
				\centering
				Classical data: \\ Input $X$ \\ Output~$Z$
			\end{minipage}
		};		
		\node[draw, rectangle, rounded corners=2mm, minimum width=2.5cm, minimum height=2cm, scale=.8]{
			\begin{minipage}{2.5cm}
				\centering 
				Encoded input: \\ $X\mapsto\rho^X$
			\end{minipage}
		};
		\node[draw, rectangle, rounded corners=2mm, minimum width=2.5cm, minimum height=2cm, scale=.8] at (3.25,0) {
			\begin{minipage}{2.5cm}
				\centering 
				Quantum computing: \\ $\mathcal{N}$
			\end{minipage}
		};
		\draw[->,xshift=3.25cm] (-1.88,.7) -- (-1.37,.7);
		\draw[->,xshift=3.25cm] (-1.88,.2) -- (-1.37,.2);
		\draw[->,xshift=3.25cm] (-1.88,-.7) -- (-1.37,-.7);
		\node[xshift=3.13cm,scale=.8] at (-2.3,-.13) {$\vdots$};
		\node[draw, rectangle, rounded corners=2mm, minimum width=2.5cm, minimum height=2cm, scale=.8] at (6.5,0) {
			\begin{minipage}{2.5cm}
				\centering
				Measure class: \\
				POVM $\{O_c\}_c$
			\end{minipage}
		};
		\draw[->,xshift=6.5cm] (-1.88,.7) -- (-1.37,.7);
		\draw[->,xshift=6.5cm] (-1.88,.2) -- (-1.37,.2);
		\draw[->,xshift=6.5cm] (-1.88,-.7) -- (-1.37,-.7);
		\node[xshift=5.72cm,scale=.8] at (-2.3,-.13) {$\vdots$};
		\node[draw, rectangle, rounded corners=2mm, minimum width=2.5cm, minimum height=2cm, scale=.8] at (9.65,0) {
		};
		\node[scale=.8] at (9.65,0) {
		\begin{minipage}{2.3cm}
			\centering
			Classical processing: \\[.3em]
			Output $\widehat{Z}$
		\end{minipage}
		};
		\draw[->,xshift=9.75cm] (-1.88,0) -- (-1.37,0);
	\end{tikzpicture}
    \caption{Information processing path for quantum statistical inference.}
    \label{fig:schematic}
\end{figure*}
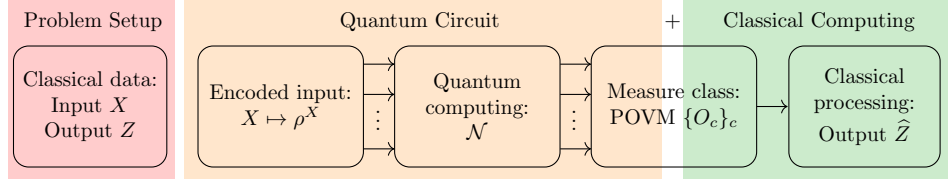

\section{Quantum-Assisted Statistical Inference}
\label{sec:inference}
Consider jointly distributed discrete random variables $X\in\X$ (input) and
$Z\in\Z$ (output), with the inference goal of predicting $Z$ from $X$. This is a general framework capturing, as an example, quantum machine learning (for classification) and quantum sensing (in discretized form). The information pipeline is depicted in Figure~\ref{fig:schematic} with its building elements discussed below. 
A \emph{quantum inference procedure} $(\Rho,\N,\F,\gamma)$ comprises:
\begin{itemize}
\item \textit{Encoding}: For each realization $X=x$, prepare system $A$ in state $\rho^x\in\mS(\H)$ with $\dim(\H)=d$, giving encoding $\Rho=\{\rho^x\}_{x\in\X}$;
\item \textit{Processing}: Apply quantum channel $\N:\mS(\H)\to\mS(\H')$, which can be a quantum machine learning policy with parameters to be optimized or a specific algorithm, such as quantum Fourier transform;
\item \textit{Measurement}: Apply POVM $\F=\{F_y\}_{y\in\Y}$ on system $A$, yielding outcome $Y\in\Y$ with probability $\mathbb{P}\{Y=y|X=x\}=\trace(F_y\N(\rho^x))$;
\item \textit{Classical post-processing}: Form estimate $\widehat{Z}\in\Z$ from $Y$ via any  stochastic kernel $\gamma_{zy}=\mathbb{P}\{\widehat{Z}=z|Y=y\}$, such as classical machine learning or signal processing algorithms.
\end{itemize}
The accuracy of the procedure is $\mathbb{P}\{\widehat{Z}=Z\}$. This is the probability of correct inference of output $Z$ based on input $X$. In what follows, we use the notation $N:=|\X|$ and $d:=\dim(\H)$ for notational brevity when needed.

\begin{theorem}[Universal Accuracy Bound]
\label{tho:stat_inference}
The accuracy of any quantum inference procedure $(\Rho,\N,\F,\gamma)$ satisfies
\begin{align}\label{eqn:inequality}
    \mathbb{P}\{\widehat{Z}=Z\}
    \leq 2^{\mathcal{Q}(X\!\to\!A)_\rho}
    \max_{z\in\Z}\mathbb{P}\{Z=z\}.
\end{align}
This bound is tight in the sense that there exists inference output $Z$, quantum processing circuit $\mathcal{N}$, POVM $\mathcal{F}$, and classical post-processing algorithm $\gamma$ for which the equality holds.     
\end{theorem}

\begin{proof}
By the definition of quantum maximal leakage and~\cite[Theorem~1]{Farokhi_PRA}, we get
\begin{align*}
    \frac{\mathbb{P}\{\widehat{Z}=Z\}}{\max_{z}\mathbb{P}\{Z=z\}}
    &\leq
    \sup_{\{F_y\}_{y\in\Y}}\sup_{Z,\widehat{Z}}
    \frac{\mathbb{P}\{\widehat{Z}=Z\}}{\max_z\mathbb{P}\{Z=z\}}\\
    &=
    2^{\mathcal{Q}(X\!\to\!A)_{\N(\rho)}}.
\end{align*}
The post-processing inequality gives
$\mathcal{Q}(X\!\to\!A)_{\N(\rho)}\leq\mathcal{Q}(X\!\to\!A)_\rho$,
establishing \revise{Eq.}~\eqref{eqn:inequality}.
The tightness of the bound stems from that the bound is attained with  $\N=\mathrm{id}$ (identity map), the optimal POVM in~\cite[Theorems~1 and~2]{Farokhi_PRA},
and the choice of $Z$ from~\cite[Theorem~1]{issa2019operational}.\end{proof}

\revise{
The quadruple $(\Rho,\N,\F,\gamma)$ that demonstrates tightness of Theorem~\ref{tho:stat_inference} is computable. The channel is the identity, the POVM is numerically computable via the algorithm in~\cite{Farokhi_PRA}, and the choice of $Z$ and $\gamma$ can be extracted from the construction in~\cite[Theorem~1]{issa2019operational} as a function of $\mathbb{P}\{Y=y|X=x\} =\trace(F_y\rho^x)$ and $\mathbb{P}\{X=x\}$.
}

\begin{remark}
The factor $2^{\mathcal{Q}(X\!\to\!A)_\rho}$ quantifies the multiplicative improvement in accuracy over the best constant estimator $\widehat{Z}=z^*:= \argmax_z\mathbb{P}\{Z=z\}$. The best constant estimator is the best policy that guesses $Z$ without any side information, i.e., without access to measurements of $X$ via quantum or classical measurements. This is referred to as the maximum \textit{a priori} estimator. This term is an indicator of the `difficulty' of inference problem in general. Importantly, $\mathcal{Q}(X\!\to\!A)_\rho$ depends on the encoding $\Rho$ but not on the inference objective $Z$ nor on the joint distribution $\mathbb{P}_{X,Z}$. This universality motivates maximizing $\mathcal{Q}(X\!\to\!A)_\rho$ over $\Rho$ to find the `best' quantum encoding policy.
\end{remark}

\begin{remark}[Why Maximizing an Upper Bound Identifies an Optimal Encoding]
\label{remark:upper_bound_justification}
\begingroup\color{black}
A natural objection is that Theorem~\ref{tho:stat_inference} supplies only an \emph{upper} bound on accuracy, and that maximizing an upper bound need not identify an encoding that is actually useful for any particular inference task. This concern is resolved by the tightness of Theorem~\ref{tho:stat_inference}. That is, for every fixed encoding $\Rho$, there exists at least one inference problem, i.e., a choice of output $Z$ together with a matching optimal quantum channel, POVM, and classical post-processing rule, for which the bound in \revise{Eq.}~\eqref{eqn:inequality} is attained with equality. Consequently, $2^{\mathcal{Q}(X\to A)_\rho}$ is not a slack ceiling; it is the exact best-case multiplicative gain in accuracy that the encoding $\Rho$ can deliver on its hardest task. Maximizing $\mathcal{Q}(X\to A)_\rho$ over $\Rho$ therefore maximizes the best accuracy that any encoding can guarantee across all possible inference tasks. An encoding with small $\mathcal{Q}(X\to A)_\rho$ cannot outperform the maximum a priori estimator by more than a factor $2^{\mathcal{Q}(X\to A)_\rho}$ on \emph{any} task, so a small value certifies a hard ceiling on every conceivable use of that encoding. This is the sense in which maximal quantum leakage identifies a universally optimal encoding, rather than an arbitrarily loose relaxation whose optimization carries no operational guarantee.
\endgroup
\end{remark}

\paragraph{Equivalence with minimum-error discrimination.}

The following theorem recasts the performance bound in Theorem~\ref{tho:stat_inference} in terms of quantum state discrimination. A version of this result is established in~\cite{xiao2026maximal}. Here, we give an independent self-contained proof for the uniform-prior case relevant here.

\begin{theorem}[Maximal Leakage As State Discrimination]
\label{thm:disc}
Let $q_x=1/N$ for all $x\in\X$ (uniform prior).
Then
\begin{align}\label{eqn:disc_equiv}
    \mathcal{Q}(X\!\to\!A)_\rho
    \;=\;
    \log \bigl(N\cdot P_{\rm guess}(\{1/N,\rho^x\}_{x\in\X})\bigr),
\end{align}
where $P_{\rm guess}(\{1/N,\rho^x\}_{x\in\X})$, defined in \revise{Eq.}~\eqref{eqn:pguess}, is the optimal success probability in minimum-error state discrimination with states $\{\rho^x\}_{x\in\X}$ and equal priors.
\end{theorem}

\begin{proof}
\medskip\noindent
\textit{Step 1: LHS $\leq$ RHS in \revise{Eq.}~\eqref{eqn:disc_equiv}.} From Definition~\ref{def:qml}, $$2^{\mathcal{Q}(X\!\to\!A)_\rho}= \sup_{\{F_y\}} \sum_{y} \max_{x} \trace(\rho^x F_y).$$ Let $\{F_y\}$ be an arbitrary POVM.
For each outcome $y$, define the decision rule
$\delta(y)\in\argmax_x\trace(\rho^x F_y)$.
Group the POVM elements by decision:
$M_x:=\sum_{y:\delta(y)=x}F_y$.
Then $\{M_x\}_{x\in\X}$ is a valid POVM, and
\begin{align*}
    \sum_y\max_x\trace(\rho^x F_y)
    &= \sum_y\trace(\rho^{\delta(y)}F_y)\\
    &= \sum_x\trace\!\Bigl(\rho^x\sum_{y:\delta(y)=x}F_y\Bigr)\\
    &= \sum_x\trace(\rho^x M_x).
\end{align*}
Hence $\sum_y\max_x\trace(\rho^x F_y)=N\cdot(1/N)\sum_x\trace(\rho^x M_x)\leq N\cdot P_{\rm guess}(\{1/N,\rho^x\}_{x\in\X})$.

\medskip\noindent
\textit{Step 2: RHS $\leq$ LHS in \revise{Eq.}~\eqref{eqn:disc_equiv}.} Let $\{M_x\}_{x\in\X}$ be any POVM feasible for \revise{Eq.}~\eqref{eqn:pguess}.
Treat each $M_x$ as a single-outcome POVM element indexed by $y=x$.
Then
\begin{align*}
    \sum_y\max_{\tilde{x}} \trace(\rho^{\tilde{x}}F_y)
    \;\geq\;
    \sum_x\trace(\rho^x M_x),
\end{align*}
so $2^{\mathcal{Q}(X\!\to\!A)_\rho}\geq N P_{\rm guess}(\{1/N,\rho^x\}_{x\in\X})$.

\medskip
Combining both steps yields the equivalence in \revise{Eq.}~\eqref{eqn:disc_equiv}.
\end{proof}

\revise{
\begin{remark}[Optimal Measurements]
Theorem~\ref{thm:disc} establishes the equivalence between  maximal quantum leakage and minimum-error discrimination with equal priors with the proof establishing a relationship between their corresponding POVMs. Rank-one POVMs also suffice for computing maximal quantum leakage~\cite{Farokhi_PRA}. For basis encoding in Proposition~\ref{prop:index}, the optimal POVM is projective and informationally incomplete. For tight frame encodings in Lemma~\ref{thm:tight}, self-referential POVMs are also rank-one, non-projective, and overcomplete. The optimal POVMs for tight frame encodings are informationally complete only if $N\geq d^2$. This holds for SIC-POVMs but not for the qubit trine.
\end{remark}

Theorem~\ref{thm:disc} is an exact characterization of quantum maximal leakage in terms of state discrimination, which requires solving a semi-definite program. The next lemma gives a closed-form bound depending only on $d$ and $N$. This bound improves the existing bound $\mathcal{Q}(X\to A)_\rho\leq\min\{\log N,2\log d\}$ of~\cite[Proposition~2]{Farokhi_PRA}.
}

\begin{lemma}[\revise{Leakage Bound}]
\label{thm:fund_bound}
For any state encoding $\{\rho^x\}_{x\in\X}$ with $N\geq d$:
\begin{align}\label{eqn:fund_bound}
    \mathcal{Q}(X\!\to\!A)_\rho \;\leq\; \log (d).
\end{align}
\end{lemma}

\begin{proof}
By Theorem~\ref{thm:disc}, it suffices to show $P_{\rm guess}(\{1/N,\rho^x\}_{x\in\X})\leq d/N$.
For any POVM $\{M_x\}_{x\in\X}$ and any  state $\rho^x$,
$\trace(\rho^x M_x)\leq \trace(\rho^x)\trace(M_x)=\trace(M_x)$, where the inequality follows from the trace relationship $\trace(AB)\leq \trace(A)\trace(B)$ for positive semi-definite operators~\cite{SHEBRAWI_ALBADAWI_2013}. Summing over $x$ and using $\sum_x\trace(M_x)=\trace(I)=d$, we get
\begin{align*}
    P_{\rm guess}(\{1/N,\rho^x\}_{x\in\X})
    &= \frac{1}{N}\sum_x \trace(\rho^x M_x)\\
    &\leq \frac{1}{N}\sum_x\trace(M_x)\\
    &= \frac{d}{N}.
\end{align*}
Hence $\mathcal{Q}(X\!\to\!A)_\rho =\log(NP_{\rm guess}(\{1/N,\rho^x\}_{x\in\X}))\leq\log (d)$.
\end{proof}

\begin{remark} \label{remark:tighter}
Lemma~\ref{thm:fund_bound} sharpens the existing bound
$\mathcal{Q}\leq\min\{\log (N),2\log (d)\}$~\cite[Proposition~2]{Farokhi_PRA} to $\mathcal{Q}\leq\min\{\log (N),\log (d)\}$. 
\end{remark}

\begin{corollary}[Number of Needed Qubits]
\label{cor:stats_inference}
For any quantum inference procedure with pure-state encoding,
\begin{align}
    \mathbb{P}\{\widehat{Z}=Z\}
    \leq
    \min \{N,d\}\max_{z}\mathbb{P}\{Z=z\}.
\end{align}
This bound is tight in the same sense as in Theorem~\ref{tho:stat_inference}.
\end{corollary}

\begin{proof}
The proof follows from substituting the bound $\mathcal{Q}(X\!\to\!A)_\rho\leq\min\{\log( N),\log (d)\}$ in Remark~\ref{remark:tighter}
into Theorem~\ref{tho:stat_inference}. The tightness stems from tightness of Theorem~\ref{tho:stat_inference} and the saturation of the bound in Remark~\ref{remark:tighter} for pure state encoding with orthogonal states. 
\end{proof}

Note that, in Corollary~\ref{cor:stats_inference}, $\dim(\H)=d$ captures the dimension of the quantum system used for statistical inference. The number of utilized qubits, if a qubit arrangement is used to create this space, is $\lceil \log(d)\rceil$.  If $d<N$, the upper bound in Corollary~\ref{cor:stats_inference} is unnecessarily reduced by the dimension of the quantum system. The tightness of Corollary~\ref{cor:stats_inference} implies that there exists, at least, one inference problem for which the performance can be improved by increasing the dimension of the underlying quantum system. This points to that the minimum number of qubits required for accurately solving a generic inference problem must be above $\log(N)$. We are not asserting that $\log(N)$ is the optimal number of required qubits, but that this is a lower bound for how many qubits are needed to solve the most `complicated' inference problems effectively (in the sense that the inference quality cannot be improved by increasing the dimension of the underlying Hilbert space). Interestingly, the only thing that matters, in this observation, is the size of the support set of the input $X$ (not its distribution, not the output $Z$, not the quantum computing method used, and not the classical post-processing procedure implemented). Therefore, this observation is rather universal.

\section{Optimal Universal Encoding}
\label{sec:optimal}
The upper bound in Theorem~\ref{tho:stat_inference}, which is a function of the maximal quantum leakage $\mathcal{Q}(X\rightarrow A)_{\rho}$, only depends on the quantum encoding of the classical data denoted by $\Rho$. This bound is also tight in the sense that it is saturated for at least one inference task. Therefore, maximizing $\mathcal{Q}(X\rightarrow A)_{\rho}$ provides a good universal encoding policy. This encoder can unlock the barrier in achieving a high accuracy in quantum-assisted statistical inference by increasing the upper bound in Theorem~\ref{tho:stat_inference}. Maximizing the upper bound in \revise{Eq.}~\eqref{eqn:inequality}, via maximizing $\mathcal{Q}(X\rightarrow A)_{\rho}$, does not make the bound looser as this bound is always attained for at least one inference problem. The universal optimal encoder is given by
\begin{align}\label{eqn:max_encoding}
    \argmax_{\rho^x\in \mS(\H),\forall x\in\X}
    \mathcal{Q}(X\!\to\!A)_\rho.
\end{align}
In the next proposition, we prove that this optimization problem attains its maximum over pure states. This is an important revelation as most quantum computing platforms and procedures rely on pure states. 

\begin{proposition}[Pure States Are Optimal]
\label{prop:maximal}
The maximum of \revise{Eq.}~\eqref{eqn:max_encoding} is attained by a pure-state encoding.
\end{proposition}

\begin{proof}
Because $\log(\cdot)$ is strictly increasing, \revise{Eq.}~\eqref{eqn:max_encoding} is
equivalent to maximizing
$g(\{\rho^x\}_{x\in\X}):=\sup_{\{F_y\}}\sum_y\max_x\trace(\rho^x F_y)$.
\revise{Let $\{\rho^x\}_{x\in\mathcal{X}}$ and $\{\sigma^x\}_{x\in\mathcal{X}}$ be any two encodings in $\mathcal{S}(\mathcal{H})^N$, and let $\alpha\in[0,1]$ be an arbitrary scalar. Then $\{\alpha\rho^x+(1-\alpha)\sigma^x\}_{x\in\mathcal{X}}$ is also an encoding.} It is easy to see that $g:\mS(\H)^{N}\rightarrow \mathbb{R}$ is convex because
    \begin{align*}
        g(\{\alpha &\rho^x+(1-\alpha)\sigma^x\}_{x\in\X})
        \\&= \sup_{\{F_y\}_{y\in\Y}}\sum_{y\in\Y} \max_{
         x\in\X} \trace((\alpha\rho^x+(1-\alpha)\sigma^x) F_y) 
         \\&\leq \alpha\sup_{\{F_y\}_{y\in\Y}}\sum_{y\in\Y} \max_{
         x\in\X} \trace(\rho^x F_y)\\
         &\quad +\!(1\!-\!\alpha)\!\!\sup_{\{F_y\}_{y\in\Y}}\sum_{y\in\Y} \max_{
         x\in\X} \trace(\sigma^x F_y) 
          \\&\leq \alpha g(\{\rho^x\}_{x\in\X})
         +(1-\alpha)g(\{\sigma^x\}_{x\in\X}).
    \end{align*}
\revise{
Also, note that, for any POVM, 
\begin{align*}
    |\max_x\mathrm{tr}(\rho^xF_y)-\max_x\mathrm{tr}(\sigma^xF_y)|
    &\leq\sum_x|\mathrm{tr}((\rho^x-\sigma^x)F_y)|\\
    &\leq\sum_x\mathrm{tr}(|\rho^x-\sigma^x|F_y).
\end{align*}
By summing over $y$ while noting that 
\begin{align*}
    \Big|\sup_{\{F_y\}}&\sum_y\max_x\trace(\rho^x F_y)-\sup_{\{F_y\}}\sum_y\max_x\trace(\sigma^x F_y)\Big|
    \\&\leq \sup_{\{F_y\}}\Big|\sum_y\max_x\trace(\rho^x F_y)-\sum_y\max_x\trace(\sigma^x F_y)\Big|
\end{align*}
results in $|g(\{\rho^x\})-g(\{\sigma^x\})|\leq\sum_x\|\rho^x-\sigma^x\|_1$. This shows that $g$ is continuous. Since $\mS(\H)^{N}$ is compact and convex and $g$ is convex and continuous,} Bauer's maximum principle~\cite[Theorem~3.5.29]{denkowski2003introduction}, originally proved in~\cite{bauer1958minimalstellen},
$g$ attains its maximum at an extreme point of $\mS(\H)^{N}$.
The extreme points of $\mS(\H)$ are the pure states~\cite[Theorem~2.3]{Holevo2013}\revise{, and the extreme points of the product $\mS(\H)^{N}$ are the N-tuples of pure states.}
\end{proof}

Note that Proposition~\ref{prop:maximal} does not claim that the solution is unique. The problem \revise{Eq.}~\eqref{eqn:max_encoding} might admit several solutions but at least one of those solutions involves pure states for encoding classical data. All the optimal solutions have the same maximal quantum leakage.

\begin{corollary}[Optimal Encoding Reformulation]
\label{cor:disc}
Maximizing $\mathcal{Q}(X\!\to\!A)_\rho$ over the encoding $\{\rho^x\}_{x\in\X}$, formulated in \revise{Eq.}~\eqref{eqn:max_encoding}, is equivalent
to designing pure-state encodings  $\{\rho^x\}_{x\in\X}$ that maximize the minimum-error discrimination success probability with equal priors.
\end{corollary}

When considering pure state encodings $\mathcal{R}=\{\rho^x\}_{x\in\X}$ with $\rho^x=\ket{\psi_x}\bra{\psi_x}$, $\forall x\in\X$, with slight abuse of notation, we refer to $\{\ket{\psi_x}\}_{x\in\X}$ as the state encoding. We consider several examples achieving the optimal encoding.

\begin{proposition}[Optimality of Basis Encoding]
\label{prop:index}
If $d\geq N$, the maximum of \revise{Eq.}~\eqref{eqn:max_encoding} is $\log (N)$, attained by the \emph{basis encoding} $\{\ket{\tau(x)}\}_{x\in\X}$, where $\{\ket{i}\}_{i=0,\dots,d-1}$ is any orthonormal basis \textcolor{black}{in} $\H$ and
$\tau:\X\to\{0,\dots,N-1\}$ is any injective map.
\end{proposition}

\begin{proof}
Note that $\mathcal{Q}(X\rightarrow A)_\rho\leq \log(N)$ \textcolor{black}{regardless} of $\{\rho^x\}_{x\in\X}$~\cite[Proposition~2]{Farokhi_PRA}. Let $\rho^x=\ket{\tau(x)} \bra{\tau(x)}$ for all $x\in\X$. \revise{Fix $\mathcal{Y}=\mathcal{X}\cup\{0\}$, $F_y=\rho^y$ for $y\in\mathcal{X}$ and $F_0=I-\sum_{y\in\mathcal{X}}\rho^y$. Then $\{F_y\}_{y\in\mathcal{Y}}$ is a POVM} and $\sum_{y\in\mathcal{Y}}\max_{x}\operatorname{tr}(\rho^xF_y)\geq\sum_{y\in\mathcal{X}}\operatorname{tr}(\rho^y)=N$. Combining this with the upper bound $\mathcal{Q}(X\rightarrow A)_\rho\leq \log(N)$ shows that $\mathcal{Q}(X\rightarrow A)_\rho= \log(N)$.
\end{proof}

Basis encoding~\cite{10555535110653511068} is thus not merely a practical convenience but the provably optimal universal encoder
whenever the Hilbert space is large enough to accommodate orthogonal codewords. The interesting and more subtle case is $N>d$, to which we now turn. The following lemma transforms the problem of finding an optimal encoding to an algebraic condition. 

\begin{lemma}[Optimality of Tight Frames]
\label{thm:tight} Assume $N\geq d$.
Let $\{\ket{\psi_x}\}_{x\in\X}$ in $\H$ be a unit-norm tight frame, i.e., $S := \sum_{x\in\X} \ket{\psi_x}\bra{\psi_x}=NI_{d}/d $. Consider state encoding $\rho^x=\ket{\psi_x}\bra{\psi_x}$ for all $x\in\X$. Then, $P_{\rm guess}(\{1/N,\rho^x\}_{x\in\X})=d/N$ and $\mathcal{Q}(X\!\to\!A)_\rho=\log (d)$.
\end{lemma}

\begin{proof} For the tight frame with $S=(N/d)I_{d}$, $\{M_x^*\}_{x\in\X}$ with $M_x^*=(d/N)\ket{\psi_x}\bra{\psi_x}$
forms a POVM because
\begin{align*}
    \sum_x M_x^* = \frac{d}{N}\sum_x|\psi_x\rangle\langle\psi_x|=\frac{d}{N}\cdot\frac{N}{d}I_{d}=I_{d},\quad M_x^*\geq0.
\end{align*}
The success probability is $(1/N)\sum_x\langle\psi_x|M_x^*|\psi_x\rangle=(1/N)(d/N)\sum_x(\bra{\psi_x}\ket{\psi_x})^2=d/N.$
Since this equals the upper bound in Lemma~\ref{thm:fund_bound}, it must be the optimal probability. The rest follows from Theorem~\ref{thm:disc}.
\end{proof}

\begin{remark}[Average State Is Maximally Mixed]
\label{rem:avg}
A tight frame satisfies $(1/N)\sum_{x\in\X}\rho^x=(1/N)\sum_x\ket{\psi_x}\bra{\psi_x}=I/d$,
meaning the ensemble-averaged state, under uniform prior, is the maximally mixed state.
\end{remark}

\begin{proposition}[Optimality of Phase Encoding]
\label{prop:phase}
If $N\geq d$, the maximum of \revise{Eq.}~\eqref{eqn:max_encoding} is $\log (d)$, attained by the \emph{phase encoding} $\{\ket{\psi_x}\}_{x\in\X}$, where $\ket{\psi_x} = \frac{1}{\sqrt{d}}\sum_{j=0}^{d-1} e^{2\pi i \tau(x)j/N}\,\ket{j}$ and $\tau:\X\to\{0,\dots,N-1\}$ is any injective map.
\end{proposition}
 
\begin{proof}
Without loss of generality assume that $\X:=\{0,\dots,N-1\}$ (by utilizing the injective map $\tau$). 
The $(j,k)$-th entry of $S$ is
\begin{align}
    S_{jk}
    &= \sum_{x=0}^{N-1}
      \frac{e^{2\pi ixj/N}\,e^{-2\pi ixk/N}}{d}\nonumber\\
    &= \frac{1}{d}\sum_{x=0}^{N-1} e^{2\pi ix(j-k)/N}\nonumber\\
    &=
    \begin{cases}
        N/d, & j=k,\\
        0, & j\neq k,
    \end{cases}
    \label{eqn:phase_Sjk}
\end{align}
where the last equality follows from discrete orthogonality of complex exponentials. This shows that $\{\ket{\psi_x}\}_{x\in\X}$ is a tight frame with $S := \sum_{x\in\X}\ket{\psi_x}\bra{\psi_x} = (N/d)\,I_d.$ The frame is unit-norm since $\bra{\psi_x}\ket{\psi_x} 
= \frac{1}{d}\sum_{j=0}^{d-1}|e^{2\pi ixj/N}|^2 = \frac{1}{d}\cdot d = 1$.
Hence $\{\ket{\psi_x}\}$ is a unit-norm tight frame with bound $N/d$. Lemma~\ref{thm:tight} gives
$\mathcal{Q}(X\!\to\!A)_\rho = \log (d)$.
\end{proof}
  
\begin{remark}[Phase Encoding and Quantum Fourier Transform]
\label{rem:phase_QFT}
The codewords $\ket{\psi_x}$ in Proposition~\ref{prop:phase} are the columns
of the $d\times N$ submatrix of the $N\times N$ Discrete Fourier Transform (DFT) matrix (the first $d$ rows). When $N=d$, the full $d\times d$ DFT matrix is unitary and phase encoding reduces to an orthonormal basis. In quantum computing, this follows from applying the quantum Fourier transform to the computational basis. For $N>d$, the $d\times N$ DFT submatrix has orthogonal rows, and its
columns (the codewords) form the tight frame. Partial Fourier matrices, obtained by selecting a subset of rows of the DFT matrix, are standard in compressed sensing \cite{candes2006robust}.
\end{remark}

\revise{
\begin{remark}[Lower Bounds]
Two lower bounds complement the upper bound $\mathcal{Q}(X\to A)_\rho\leq \min\{\log (N),\log (d)\}$ coming from \cite[Proposition~2]{Farokhi_PRA} and Lemma~\ref{thm:fund_bound}. First, consider a pure-state encoding $\rho^x=|\psi_x\rangle\langle\psi_x|$, $x\in\X$, whose frame operator $S=\sum_{x\in\mathcal{X}}|\psi_x\rangle\langle\psi_x|$ is invertible. The measurement operators $M_x=S^{-1/2}|\psi_x\rangle\langle\psi_x|S^{-1/2}$ form a feasible POVM in Eq.~\eqref{eqn:pguess}. Therefore, Theorem~\ref{thm:disc} gives
\begin{equation*}
\mathcal{Q}(X\to A)_\rho\geq \log\left(\sum_{x\in\mathcal{X}}\langle\psi_x|S^{-1/2}|\psi_x\rangle^2 \right).
\end{equation*}
For tight frames, where $S^{-1/2}=\sqrt{d/N}\,I$, the right-hand side equals $\log (d)$, which coincides with the upper bound $\log (d)$ for the optimal encoding identified in Lemma~\ref{thm:tight}. Second, let $I_{\rm acc}(\{\rho^x,P_X(x)\}_{x\in\mathcal{X}})$ be the accessible information, that is, the maximum over POVMs $\{F_y\}_{y\in\mathcal{Y}}$ of the Shannon mutual information $I(X;Y)$ between $X$ and the outcome $Y$, with $P\{Y=y|X=x\}=\mathrm{tr}(F_y\rho^x)$ and $P_X(x)=P\{X=x\}>0$. In~\cite[Proposition 4]{Farokhi_PRA}, it was shown that
\begin{equation*}
\mathcal{Q}(X\to A)_\rho\geq I_{\rm acc}(\{\rho^x,P_X(x)\}_{x\in\mathcal{X}}).
\end{equation*}
The optimal encoding in Proposition~\ref{prop:index} gives $I_{\rm acc}(\{\rho^x,P_X(x)\}_{x\in\mathcal{X}})=H(X)$, where $H(X)$ is the Shannon entropy of $X$. Since $\mathcal{Q}(X\to A)_\rho=\log (N)$ in this case, the lower bound is saturated if and only if $X$ is uniformly distributed.
\end{remark}
}

Phase encoding is not the only tight frame. Therefore, the optimal encoding is not unique. Among all tight frames achieving $\mathcal{Q}(X\to A)_{\rho}=\log (d)$, some are more symmetric than others. In what follows, we consider pairwise similarity between codewords and present symmetric constructions when possible. 

\section{Equiangular Tight Frames}
\label{sec:constructions}

\subsection{The Welch bound}
The natural measure of pairwise similarity between codewords $\ket{\psi_x}$ and $\ket{\psi_{x'}}$, $x\neq x'$, is the squared overlap $|\bra{\psi_x}\ket{\psi_{x'}}|^2$, which is known as the pure state fidelity in quantum information theory~\cite[\chap\,9.2]{wilde2013quantum}. The following theorem gives a fundamental lower bound on the worst-case overlap. The proof is based on the seminal work of Welch in information theory in~\cite{welch1974lower} and is presented here for the sake of completeness. 

\begin{theorem}[Welch Bound For Tight Frames]
\label{thm:welch}
Assume $N>d$. If $\{\ket{\psi_x}\}_{x\in\X}$ is a unit-norm tight frame in $\H\cong\mathbb{C}^d$, then
\begin{align}\label{eqn:welch}
    \max_{x\neq x'}|\bra{\psi_x}\ket{\psi_{x'}}|^2
    \;\geq\;
    \frac{N-d}{d(N-1)}.
\end{align}
\end{theorem}

\begin{proof}
Noting $S=\sum_x \ket{\psi_x}\bra{\psi_x}=(N/d)I$, we get
\begin{align*}
    \frac{N^2}{d}=\trace(S^2)&= \sum_{x,x'}|\bra{\psi_x}\ket{\psi_{x'}}|^2
    = N + \sum_{x\neq x'}|\bra{\psi_x}\ket{\psi_{x'}}|^2.
\end{align*}
Hence $\sum_{x\neq x'} |\bra{\psi_x}\ket{\psi_{x'}}|^2=N(N-d)/d$.
Since there are $N(N-1)$ ordered pairs $(x,x')$ with $x\neq x'$:
\begin{align*}
    \max_{x\neq x'}|\bra{\psi_x}\ket{\psi_{x'}}|^2
    \;\geq\;
    \frac{N(N-d)/d}{N(N-1)}
    = \frac{N-d}{d(N-1)}.
\end{align*}
\end{proof}

\begin{remark}
    The proof of Theorem~\ref{thm:welch} shows that the average squared overlap for tight frames is always equal to 
    \begin{align*}
        \frac{1}{N(N-1)}\sum_{x\neq x'} |\bra{\psi_x}\ket{\psi_{x'}}|^2=\frac{N-d}{d(N-1)}.
    \end{align*}
    The average squared overlap tends to $(N-d)/[d(N-1)] \to 1/d$
     as $N \to \infty$.
\end{remark}

Equality in Theorem~\ref{thm:welch} holds if and only if all pairwise squared overlaps are equal:
$|\langle\psi_x|\psi_{x'}\rangle|^2=c^2:=(N-d)/[d(N-1)]$ for all $x\neq x'$. This motivates the following definition.

\begin{definition}[Equiangular Tight Frame (ETF)]
\label{def:ETF}
A unit-norm tight frame $\{\ket{\psi_x} \}_{x\in\X}$, for $N\geq d$, is an \emph{equiangular tight frame} (ETF) with parameters $(N,d)$ if it
saturates the Welch bound \revise{Eq.}~\eqref{eqn:welch}, i.e., 
\begin{align}\label{eqn:etf_coh}
    |\bra{\psi_x}\ket{\psi_{x'}}| = c_{N,d}:=\sqrt{\frac{N-d}{d(N-1)}}
    \quad \forall\; x\neq x'.
\end{align}
The quantity $c_{N,d}$ is referred to as the \emph{coherence} of the ETF.
\end{definition}

\revise{Equiangular lines and equiangular tight frames have been studied in~\cite{lemmens1973, strohmer2003grassmannian, sustik2007}, and their measurement counterparts, SIC-POVMs, in~\cite{zauner2011quantum,renes2004symmetric}. Finding tight frames with small coherence is also related to spherical codes~\cite{conway1996}.}

ETFs are optimal quantum encodings in two complementary senses. They achieve the maximum leakage $\mathcal{Q}(X\to A)_\rho=\log (d)$ (tight unit-norm frame condition), and they simultaneously minimize the maximum pairwise overlap among all tight frames
(Welch bound saturation). A small coherence $c_{N,d}$ means the codewords are as ``spread out'' as possible in the Hilbert space $\H\cong\mathbb{C}^d$, which in turn means they are maximally distinguishable in a
minimax sense.

Table~\ref{tab:etf} summarises the key special cases of ETF encodings.
When $N=d$ the ETF degenerates to an orthonormal basis (the complete regime of
Proposition~\ref{prop:index}).
As $N$ increases past $d$, the coherence $c_{N,d}$ increases from $0$
(orthogonal codewords). We now provide explicit codeword sets for these parameter regimes.
Existence of an ETF$(N,d)$ is not guaranteed for all $(N,d)$. The problem is
connected to deep questions in combinatorics and algebraic number
theory~\cite{waldron2018frames, strohmer2003grassmannian}. 

\begin{table}[t]
\centering
\caption{Key ETF parameter regimes with $\mathcal{Q}(X\to A)_\rho=\log (d)$.}
\label{tab:etf}
\begin{tabular}{@{}llll@{}}
\toprule
Parameters & Structure & Coherence $c_{N,d}$ & $P_{\rm guess}$ \\
\midrule
$N=d$     & Orthonormal basis      & $0$          & $1$        \\
$N=d+1$   & Regular simplex        & $1/d$        & $d/(d+1)$  \\
$N=d^2$   & SIC-POVM               & $1/\sqrt{d+1}$ & $1/d$    \\
\bottomrule
\end{tabular}
\end{table}

\subsection{Regular simplex}

The regular simplex ETF exists for every $d\geq1$ and corresponds to $d+1$ equidistant
points on the unit sphere in $\mathbb{C}^d$, i.e., the vertices of a regular simplex
inscribed in the sphere.

\begin{figure}[t]
\algocaption{Regular simplex construction for ETF$(d{+}1,\,d)$}\label{alg:simplex}
    \vspace{2pt}\hrule\vspace{4pt}
    \begin{algorithmic}[1]
    \Require $d$
    \Ensure {$\ket{\psi_k}\in\mathbb{C}^d$, $\forall k\in\{0,\dots,d\}$},
    \State Pick $\{\ket{e_k}\}_{k=0}^{d}$ as the basis for $\mathbb{C}^{d+1}$
    \State $\ket{u}\leftarrow \frac{1}{\sqrt{d+1}}\sum_{k=0}^{d}\ket{e_k}
    \;\in\;\mathbb{C}^{d+1},$
    \State $|\tilde{\psi}_k\rangle
    \leftarrow \sqrt{\frac{d+1}{d}}\,
    (I_{d+1} - \ket{u}\bra{u})\ket{e_k}
    \;\in\; \mathbb{C}^{d+1}, k = 0,\ldots,d$
    \State Set $Q\in\mathbb{C}^{(d+1)\times d}$ as any isometry with $\operatorname{col}(Q) = u^{\perp}$, i.e., \ $Q^{\dagger}Q = I_d$ and $QQ^{\dagger} = I_{d+1} - |u\rangle\langle u|$,
    \State $\ket{\psi_k} \leftarrow Q^{\dagger}|\tilde{\psi}_k\rangle
    \;\in\; \mathbb{C}^{d},
    k = 0,\ldots,d,$
    \State \Return $\ket{\psi_k}, k = 0,\ldots,d,$
    \end{algorithmic}
    \vspace{4pt}\hrule
\end{figure}

\begin{proposition}[Regular Simplex]
\label{prop:simplex}
For every $d \geq 1$, there exists an equiangular tight frame of $N=d+1$ unit vectors in $\mathbb{C}^d$ with coherence $c_{d+1,d} = 1/d$, achieving $\mathcal{Q}(X\to A)_\rho = \log (d)$ and $P_{\rm guess}(\{1/N,\rho^x\}_{x\in\X}) = d/(d+1)$. Algorithm~\ref{alg:simplex} provides this construction.
\end{proposition}

\begin{proof}
We verify the three defining properties of unit norm, equiangularity at
$c_{d+1,d}=1/d$, and the tight frame condition $\sum_k|\psi_k\rangle\langle\psi_k|=(N/d)I_d$.

\medskip
\noindent\textbf{Step 1 (Unit norm).}
Since $|\tilde\psi_k\rangle\in u^{\perp} = \operatorname{col}(Q)$,
the map $Q^{\dagger}$ acts as an isometry on $u^{\perp}$, so
$\bra{\psi_k}\ket{\psi_k} = \langle Q^{\dagger}\tilde\psi_k | Q^{\dagger}\tilde\psi_k\rangle = \langle\tilde\psi_k|\tilde\psi_k\rangle$. Note that
\begin{align*}
    \ket{u}\bra{u} \ket{e_k}
    =\ket{u} \left(\frac{1}{\sqrt{d+1}}\sum_{\ell=0}^{d}\bra{e_\ell}\ket{e_k}\right)
    =\frac{1}{\sqrt{d+1}}\ket{u},
\end{align*}
and, as a result,
\begin{align*}
    \langle\tilde\psi_k |\tilde\psi_k\rangle
    &= \frac{d+1}{d}\,
       \bra{e_k}\,\bigl(I_{d+1}-|u\rangle\langle u|\bigr)^2\,\ket{e_k}\\
    &= \frac{d+1}{d}\left(\bra{e_k}
       -\frac{1}{\sqrt{d+1}}\bra{u}
       \right)\\
       &\qquad\times\left(\ket{e_k}
       -\frac{1}{\sqrt{d+1}}\ket{u}
       \right)\\
   &=\frac{d+1}{d} \left( 1+\frac{1}{d+1}\right)
   \\
   &\quad-\frac{\sqrt{d+1}}{d}\left(
   \bra{e_k}\ket{u}
   +
   \bra{u}\ket{e_k}
   \right)\\
   &=\frac{d+2}{d}-\frac{2}{d}\\
   &=1.
\end{align*}

\medskip
\noindent\textbf{Step 2 (Equiangularity).}
For any $k,k'$, since both $|\tilde\psi_k\rangle$ and $|\tilde\psi_{k'}\rangle$
lie in $\operatorname{col}(Q)$. Hence, the operator $QQ^{\dagger} = I_{d+1}-|u\rangle\langle u|$ acts as the identity on them, giving $\langle\psi_k| \psi_{k'}\rangle = \langle Q^{\dagger}\tilde\psi_k\, |Q^{\dagger}\tilde\psi_{k'}\rangle= \langle\tilde\psi_k| \tilde\psi_{k'}\rangle.$ For $k\neq k'$, 
\begin{align*}
    \langle\tilde\psi_k| \tilde\psi_{k'}\rangle
    &= \frac{d+1}{d}\left(\bra{e_k}
       -\frac{1}{\sqrt{d+1}}\bra{u}
       \right)\\
       &\qquad\times \left(\ket{e_{k'}}
       -\frac{1}{\sqrt{d+1}}\ket{u}
       \right)\\
    &=\frac{1}{d}-\frac{\sqrt{d+1}}{d}\left(
   \bra{e_k}\ket{u}
   +
   \bra{u}\ket{e_{k'}}
   \right)\\
   &=-\frac{1}{d}.
\end{align*}
This means $c_{d+1,d}=1/d$, confirming equiangularity at the Welch bound.

\medskip
\noindent\textbf{Step 3 (Tight frame).}
We have
\begin{align}
    \sum_{k=0}^{d} |\tilde\psi_k\rangle \langle\tilde\psi_k|
    &= \frac{d+1}{d}\,
       (I_{d+1}-\ket{u}\bra{u})
       \!\left(\sum_{k=0}^{d} \ket{e_k}\bra{e_k}\right)\notag\\
       &\qquad\times 
       (I_{d+1}-\ket{u}\bra{u})
       \notag\\
    &= \frac{d+1}{d}\,
       (I_{d+1}-\ket{u}\bra{u})
       \,I_{d+1}\notag\\
       &\qquad\times 
       (I_{d+1}-\ket{u}\bra{u})
       \notag\\
    &= \frac{d+1}{d}\,(I_{d+1}-\ket{u}\bra{u})\\
    &= \frac{d+1}{d}\,QQ^{\dagger},
     \label{eqn:simplex_frame_op_ambient}
\end{align}
where in the last step we used that the projector $(I_{d+1}-\ket{u}\bra{u})$ squares to itself and
$QQ^{\dagger} = I_{d+1}-\ket{u}\bra{u}$.
Applying $Q^{\dagger}$ from the left and $Q$ from the right:
\begin{align}
    \sum_{k=0}^{d} \ket{\psi_k} \bra{\psi_k}
    &= Q^{\dagger}\!\left(\sum_{k=0}^{d}|\tilde\psi_k\rangle\langle\tilde\psi_k|\right)\!Q\notag\\
    &= \frac{d+1}{d}\,Q^{\dagger}QQ^{\dagger}Q\notag\\
    &= \frac{d+1}{d}\,I_d\notag\\
    &= \frac{N}{d}\,I_d,
    \label{eqn:simplex_tight}
\end{align}
confirming the tight frame condition.

\medskip
Together, Steps 1--3 establish that $\{|\psi_k\rangle\}_{k=0}^{d}$ is an
ETF$(d+1,d)$. Lemma~\ref{thm:tight} then shows that the optimal POVM is the self-referential one, $M_k^* = (d/N) \ket{\psi_k}\bra{\psi_k}$, achieving
$P_{\rm guess}(\{1/N,\rho^x\}_{x\in\X}) = d/(d+1)$ and $\mathcal{Q}(X\to A)_\rho = \log (d)$.
\end{proof}

\begin{figure}
    \centering
    \tdplotsetmaincoords{72}{-20}
    \begin{tikzpicture}[
      tdplot_main_coords,
      scale=2.6,
      rotate=15, 
      >=Stealth,
      font=\small,
      line cap=round,
      line join=round,
    ]
      \pgfmathsetmacro{\sq}{0.8660254}   
     
      \colorlet{csphere}{blue!6!white}
      \colorlet{ccirc}  {blue!55!black}
      \colorlet{ceq}    {gray!50}
      \colorlet{cax}    {gray!60}
      \colorlet{cpsi0}  {red!72!black}       
      \colorlet{cpsi1}  {violet!75!black}    
      \colorlet{cpsi2}  {orange!85!black}    
      \draw[gray!30, fill=csphere] (0,0,0) circle (1);
      \tdplotdrawarc[ceq, dashed, thin]{(0,0,0)}{1}{190}{370}{}{}
      \draw[ccirc!35, dashed, thin]
        (0.95782,0,0.28736) -- (0.97176,0,0.23598) -- (0.98294,0,0.18392) -- (0.99134,0,0.13135) -- (0.99692,0,0.07840) -- (0.99968,0,0.02523) -- (0.99961,0,-0.02801) -- (0.99670,0,-0.08117) -- (0.99097,0,-0.13410) -- (0.98243,0,-0.18665) -- (0.97110,0,-0.23868) -- (0.95702,0,-0.29002) -- (0.94023,0,-0.34055) -- (0.92077,0,-0.39010) -- (0.89870,0,-0.43856) -- (0.87409,0,-0.48577) -- (0.84700,0,-0.53160) -- (0.81750,0,-0.57592) -- (0.78569,0,-0.61862) -- (0.75165,0,-0.65956) -- (0.71549,0,-0.69863) -- (0.67729,0,-0.73572) -- (0.63717,0,-0.77072) -- (0.59525,0,-0.80354) -- (0.55164,0,-0.83408) -- (0.50647,0,-0.86226) -- (0.45986,0,-0.88799) -- (0.41195,0,-0.91121) -- (0.36287,0,-0.93184) -- (0.31276,0,-0.94983) -- (0.26176,0,-0.96513) -- (0.21002,0,-0.97770) -- (0.15769,0,-0.98749) -- (0.10491,0,-0.99448) -- (0.05183,0,-0.99866) -- (-0.00139,0,-1.00000) -- (-0.05461,0,-0.99851) -- (-0.10768,0,-0.99419) -- (-0.16044,0,-0.98705) -- (-0.21274,0,-0.97711) -- (-0.26444,0,-0.96440) -- (-0.31540,0,-0.94896) -- (-0.36546,0,-0.93083) -- (-0.41448,0,-0.91006) -- (-0.46233,0,-0.88671) -- (-0.50886,0,-0.86085) -- (-0.55396,0,-0.83254) -- (-0.59748,0,-0.80188) -- (-0.63931,0,-0.76895) -- (-0.67933,0,-0.73383) -- (-0.71743,0,-0.69663) -- (-0.75349,0,-0.65746) -- (-0.78741,0,-0.61643) -- (-0.81910,0,-0.57365) -- (-0.84847,0,-0.52924) -- (-0.87544,0,-0.48333) -- (-0.89992,0,-0.43606) -- (-0.92185,0,-0.38754) -- (-0.94117,0,-0.33793) -- (-0.95782,0,-0.28736);
      \fill[orange!12, opacity=0.70]
        ({\sq},0,0.5) -- ({-\sq},0,0.5) -- (0,0,-1) -- cycle;
      \draw[gray!65, thick, dashed] ({\sq},0,0.5) -- ({-\sq},0,0.5);
      \draw[gray!55, thick, dashed] ({-\sq},0,0.5) -- (0,0,-1);
      \draw[gray!55, thick, dashed] (0,0,-1) -- ({\sq},0,0.5);
      \tdplotdrawarc[ceq, thin]{(0,0,0)}{1}{10}{190}{}{}
      \draw[ccirc, thick]
        (-0.95782,0,-0.28736) -- (-0.97176,0,-0.23598) -- (-0.98294,0,-0.18392) -- (-0.99134,0,-0.13135) -- (-0.99692,0,-0.07840) -- (-0.99968,0,-0.02523) -- (-0.99961,0,0.02801) -- (-0.99670,0,0.08117) -- (-0.99097,0,0.13410) -- (-0.98243,0,0.18665) -- (-0.97110,0,0.23868) -- (-0.95702,0,0.29002) -- (-0.94023,0,0.34055) -- (-0.92077,0,0.39010) -- (-0.89870,0,0.43856) -- (-0.87409,0,0.48577) -- (-0.84700,0,0.53160) -- (-0.81750,0,0.57592) -- (-0.78569,0,0.61862) -- (-0.75165,0,0.65956) -- (-0.71549,0,0.69863) -- (-0.67729,0,0.73572) -- (-0.63717,0,0.77072) -- (-0.59525,0,0.80354) -- (-0.55164,0,0.83408) -- (-0.50647,0,0.86226) -- (-0.45986,0,0.88799) -- (-0.41195,0,0.91121) -- (-0.36287,0,0.93184) -- (-0.31276,0,0.94983) -- (-0.26176,0,0.96513) -- (-0.21002,0,0.97770) -- (-0.15769,0,0.98749) -- (-0.10491,0,0.99448) -- (-0.05183,0,0.99866) -- (0.00139,0,1.00000) -- (0.05461,0,0.99851) -- (0.10768,0,0.99419) -- (0.16044,0,0.98705) -- (0.21274,0,0.97711) -- (0.26444,0,0.96440) -- (0.31540,0,0.94896) -- (0.36546,0,0.93083) -- (0.41448,0,0.91006) -- (0.46233,0,0.88671) -- (0.50886,0,0.86085) -- (0.55396,0,0.83254) -- (0.59748,0,0.80188) -- (0.63931,0,0.76895) -- (0.67933,0,0.73383) -- (0.71743,0,0.69663) -- (0.75349,0,0.65746) -- (0.78741,0,0.61643) -- (0.81910,0,0.57365) -- (0.84847,0,0.52924) -- (0.87544,0,0.48333) -- (0.89992,0,0.43606) -- (0.92185,0,0.38754) -- (0.94117,0,0.33793) -- (0.95782,0,0.28736);
     
      \node[ccirc, font=\scriptsize\itshape, anchor=south,rotate=25]
        at (0, 0.03, 0.9836)
        {Bloch circle ($xz$-plane)};
     
      \draw[cax, ->] (0,0,0) -- (1.40,0,0) node[right, black!65]{$x$};
      \draw[cax, ->] (0,0,0) -- (0,1.28,0) node[above right, black!65]{$y$};
      \draw[cax, ->] (0,0,0) -- (0,0,1.52) node[above, black!65]{$z$};
      \draw[gray!60, ->, thin]
        (0.27352,0,0.11051) -- (0.28281,0,0.08391) -- (0.28953,0,0.05655) -- (0.29360,0,0.02868) -- (0.29500,0,0.00054) -- (0.29371,0,-0.02760) -- (0.28973,0,-0.05549) -- (0.28312,0,-0.08287) -- (0.27392,0,-0.10950) -- (0.26223,0,-0.13513) -- (0.24814,0,-0.15953) -- (0.23179,0,-0.18247) -- (0.21333,0,-0.20375) -- (0.19292,0,-0.22317) -- (0.17076,0,-0.24056) -- (0.14703,0,-0.25575) -- (0.12196,0,-0.26861) -- (0.09579,0,-0.27902) -- (0.06873,0,-0.28688) -- (0.04106,0,-0.29213);
      \node[gray!65, font=\tiny,rotate=60] at (0.4215/2,0,-0.2434/2){$120^\circ$};
      \draw[gray!60, ->, thin]
        (-0.04106,0,-0.29213) -- (-0.06873,0,-0.28688) -- (-0.09579,0,-0.27902) -- (-0.12196,0,-0.26861) -- (-0.14703,0,-0.25575) -- (-0.17076,0,-0.24056) -- (-0.19292,0,-0.22317) -- (-0.21333,0,-0.20375) -- (-0.23179,0,-0.18247) -- (-0.24814,0,-0.15953) -- (-0.26223,0,-0.13513) -- (-0.27392,0,-0.10950) -- (-0.28312,0,-0.08287) -- (-0.28973,0,-0.05549) -- (-0.29371,0,-0.02760) -- (-0.29500,0,0.00054) -- (-0.29360,0,0.02868) -- (-0.28953,0,0.05655) -- (-0.28281,0,0.08391) -- (-0.27352,0,0.11051);
      \node[gray!65, font=\tiny,rotate=-60] at (-0.4215/2,0,-0.2434/2){$120^\circ$};
     
      \draw[gray!60, ->, thin]
        (-0.23246,0,0.18162) -- (-0.21408,0,0.20297) -- (-0.19374,0,0.22246) -- (-0.17164,0,0.23993) -- (-0.14797,0,0.25521) -- (-0.12295,0,0.26816) -- (-0.09681,0,0.27866) -- (-0.06979,0,0.28663) -- (-0.04213,0,0.29198) -- (-0.01409,0,0.29466) -- (0.01409,0,0.29466) -- (0.04213,0,0.29198) -- (0.06979,0,0.28663) -- (0.09681,0,0.27866) -- (0.12295,0,0.26816) -- (0.14797,0,0.25521) -- (0.17164,0,0.23993) -- (0.19374,0,0.22246) -- (0.21408,0,0.20297) -- (0.23246,0,0.18162);
      \node[gray!65, font=\tiny, xshift=2pt,rotate=15] at (0.0000,0,0.4867/2){$120^\circ$};
     
      \draw[->, cpsi0, line width=1.5pt] (0,0,0) -- ({\sq},0,0.5);
      \fill[cpsi0] ({\sq},0,0.5) circle (1.6pt);
      \node[cpsi0, anchor=west, font=\small, xshift=3pt]
        at ({\sq},0,0.53){$\ket{\psi_0}$};
     
      \draw[->, cpsi1, line width=1.5pt] (0,0,0) -- ({-\sq},0,0.5);
      \fill[cpsi1] ({-\sq},0,0.5) circle (1.6pt);
      \node[cpsi1, anchor=east, font=\small, xshift=-3pt]
        at ({-\sq},0,0.53){$\ket{\psi_1}$};
     
      \draw[->, cpsi2, line width=1.5pt] (0,0,0) -- (0,0,-1);
      \fill[cpsi2] (0,0,-1) circle (1.6pt);
      \node[cpsi2, anchor=north west, font=\small, xshift=3pt, yshift=-2pt]
        at (0,0,-1.07){$\ket{\psi_2}$};
    \end{tikzpicture}
    \caption{Qubit trine in Remark~\ref{remark:trine} as regular simplex for $d=2$.}
    \label{fig:trine}
\end{figure}
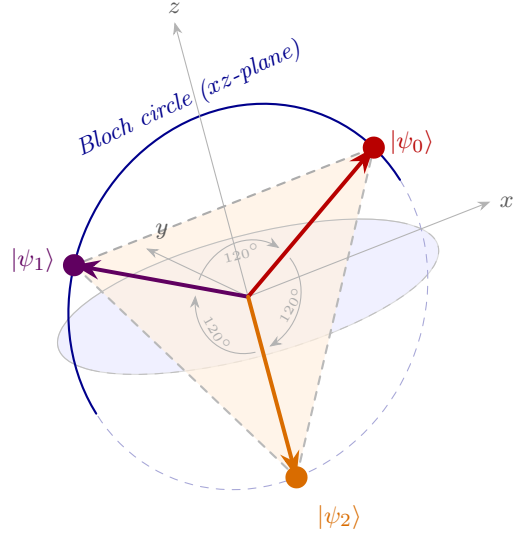

\begin{remark}[Qubit Trine as Regular Simplex for $d=2$]
\label{remark:trine}
We verify the general construction of Proposition~\ref{prop:simplex}
step by step for $d=2$, $N=3$. The centroid vector is
\begin{equation}
    |u\rangle = \tfrac{1}{\sqrt{3}}\bigl(|e_0\rangle+|e_1\rangle+|e_2\rangle\bigr)
    = \tfrac{1}{\sqrt{3}}
    \begin{bmatrix}
        1 & 1 & 1
    \end{bmatrix}
    ^\top ,
    \label{eqn:trine_u}
\end{equation}
which results in 
\begin{align}
    |\tilde{\psi}_0\rangle &= \sqrt{\tfrac{1}{6}}\begin{bmatrix}
        2 & -1 & -1
    \end{bmatrix}^\top ,
    \label{eqn:trine_psi0}\\
    |\tilde{\psi}_1\rangle &= \sqrt{\tfrac{1}{6}}\begin{bmatrix}
        -1 & 2 & -1
    \end{bmatrix}^\top ,
    \label{eqn:trine_psi1}\\
    |\tilde{\psi}_2\rangle &= \sqrt{\tfrac{1}{6}}\begin{bmatrix}
        -1 & -1 & 2
    \end{bmatrix}^\top .
    \label{eqn:trine_psi2}
\end{align}
We need $Q\in\mathbb{R}^{3\times 2}$ satisfying $Q^\dagger Q=I_2$ and
$\operatorname{col}(Q)=u^{\perp}$.
A natural orthonormal basis for $u^{\perp}$ is obtained by
Gram--Schmidt orthogonalization of the two difference vectors
$\ket{e_0}-\ket{e_1}$ and $\ket{e_1}-\ket{e_2}$:
\begin{align}
    q_1 &
         = \sqrt{\tfrac{1}{2}}\,
         \begin{bmatrix}
             1 & -1 & 0
         \end{bmatrix}^\top 
         \label{eqn:q1}\\
    q_2 &= \sqrt{\tfrac{1}{6}}
    \,
         \begin{bmatrix}
             1 & 1 & -2
         \end{bmatrix}^\top .         \label{eqn:q2}
\end{align}
The isometry is therefore $Q = \bigl[q_1\;\big|\;q_2\bigr]$.
One can verify directly that $Q^\dagger Q=I_2$ (columns are orthonormal)
and $QQ^\dagger = I_3 - \ket{u}\bra{u}$ (the orthogonal projector onto $u^{\perp}$). This results in 
\begin{align}
    \ket{\psi_0}
    &\!=\! Q^\dagger|\tilde{\psi}_0\rangle
     \!=\! \tfrac{1}{2}\begin{bmatrix}
         \sqrt{3} & 1
     \end{bmatrix}^\top 
     \!=\! \cos(\tfrac{\pi}{6})\ket{0} \!+\! \sin(\tfrac{\pi}{6})\ket{1},
     \label{eqn:trine_Cd_0}\\
    \ket{\psi_1}
    &\!=\! Q^\dagger|\tilde{\psi}_1\rangle
     \!=\! \tfrac{1}{2}\begin{bmatrix}
         -\sqrt{3} & 1
     \end{bmatrix}^\top 
     \!=\! \cos(\tfrac{5\pi}{6})\ket{0} \!+\! \sin(\tfrac{5\pi}{6})\ket{1},
     \label{eqn:trine_Cd_1}\\
    \ket{\psi_2}
    &\!=\! Q^\dagger|\tilde{\psi}_2\rangle
     \!=\! \begin{bmatrix}
        0 & -1
     \end{bmatrix}^\top 
     \!=\! \cos(\tfrac{3\pi}{2})\ket{0} \!+\! \sin(\tfrac{3\pi}{2})\ket{1}.
     \label{eqn:trine_Cd_2}
\end{align}
These three states form an equilateral triangle. More compactly, 
\begin{align*}
    \ket{\psi_k} = \cos(\tfrac{2k\pi}{3}+\tfrac{\pi}{6})\ket{0}+ \sin(\tfrac{2k\pi}{3}+\tfrac{\pi}{6})\ket{1}, \forall k=0,1,2.
\end{align*}
These states are depicted in Figure~\ref{fig:trine} in a Bloch sphere. Finally, note that the term \emph{trine} comes from the Latin \emph{trinus} (threefold) and the
astrological trine aspect of $120^{\circ}$. The double trine ensemble built from the qubit trine is studied in~\cite{peres1991optimal} for optimal quantum state detection. The qubit trine is also used in three-state quantum key distribution~\cite{phoenix2000three, boileau2005unconditional}. 
This paper shows that qubit trine is also optimal for statistical inference.
\end{remark}

\subsection{SIC-POVMs}
\label{subsec:sic}

SIC-POVMs are ETFs$(d^2,d)$ with coherence $1/\sqrt{d+1}$. A constructive approach in arbitrary dimension uses the discrete Heisenberg--Weyl group.

\begin{definition}[Heisenberg--Weyl SIC-POVM {\cite{renes2004symmetric,zauner2011quantum}}]
\label{def:sic}
Let $\omega=e^{2\pi i/d}$ and $\tau=e^{i\pi(d+1)/d}$.
Define the generalized Pauli operators on $\mathbb{C}^d$ by
\begin{align}
    Z|j\rangle = \omega^j|j\rangle,\quad
    X|j\rangle = |j\oplus 1\rangle,
\end{align}
for all $j\in\mathbb{Z}_d:=\{0,\dots,d-1\}$ with $\oplus$ denoting addition modulo $d$. The displacement operators are $D_{p,q}:=\tau^{pq}X^pZ^q$ for $p,q\in\mathbb{Z}$. A unit vector $|f\rangle\in\mathbb{C}^d$ is a {fiducial state} for a SIC-POVM if
\begin{align}\label{eqn:fiducial}
    |\langle f|D_{p,q}|f\rangle|^2 &= \frac{1}{d+1}\nonumber\\
    &\forall\;(p,q)\in\mathbb{Z}_d \times\mathbb{Z}_d,\ (p,q)\neq(0,0).
\end{align}
When a fiducial state $|f\rangle$ exists, the $d^2$ states
$|\psi_{p,q}\rangle=D_{p,q}|f\rangle$
constitute a SIC-POVM.
\end{definition}

\begin{proposition}[Heisenberg--Weyl SIC-POVMs are ETFs]
\label{prop:hw_sic_etf}
Suppose a fiducial state $|f\rangle\in\mathbb{C}^d$ satisfying \revise{Eq.}~\eqref{eqn:fiducial} exists.
Then, $N=d^2$ states
$|\psi_{p,q}\rangle = D_{p,q}|f\rangle$, $(p,q)\in\mathbb{Z}_d\times\mathbb{Z}_d$,
form an {ETF$(d^2,d)$}  with coherence $c_{d^2,d} = 1/{\sqrt{d+1}}$ achieving $P_{\rm guess}(\{1/N,\rho^x\}_{x\in\X})=1/d$ and
$\mathcal{Q}(X\!\to\!A)_\rho=\log (d)$.
\end{proposition}
 
\begin{proof}
We verify the three properties of unit norm, tight frame condition, and equiangularity.
 
\medskip
\noindent\textbf{(i) Unit norm.}
Since every displacement operator $D_{p,q}=\tau^{pq}X^p Z^q$ is unitary (both $X$ and $Z$ are unitary, and the phase $\tau^{pq}$ has magnitude one), hence $\bra{\psi_{p,q}}\ket{\psi_{p,q}}=\bra{f}D_{p,q}^\dagger D_{p,q}\ket{f}=\bra{f}\ket{f}=1$.
 
\medskip
\noindent\textbf{(ii) Tight frame.}
Define the frame operator
\begin{equation}\label{eqn:sic_frame_op}
    S := \sum_{p,q\in\mathbb{Z}_d}
    D_{p,q}|f\rangle\langle f|D_{p,q}^\dagger.
\end{equation}
\emph{Step 1 (covariance).}
The displacement operators satisfy
\begin{equation}\label{eqn:weyl}
    D_{r,s}\,D_{p,q}
    \;=\; \tau^{sp-rq}\,D_{r+p,\,s+q},
\end{equation}
a consequence of $ZX=\omega XZ$ and $\omega=\tau^2$, which give $D_{r,s}D_{p,q}=\tau^{rs+pq+2sp}X^{r+p}Z^{s+q}$. Since $X^d=Z^d=I$ and $\tau^d=\pm1$, we have $D_{p+d,q}=\pm D_{p,q}$ and $D_{p,q+d}=\pm D_{p,q}$, so $D_{p,q}|f\rangle\langle f|D_{p,q}^\dagger$ and $|\langle f|D_{p,q}|f\rangle|$ depend only on $(p,q)$ modulo $d$. Noting that $|\tau^{sp-rq}|=1$, we get
\begin{align}
    D_{r,s}\,S\,D_{r,s}^\dagger
    &= \sum_{p,q\in\mathbb{Z}_d} \bigl(D_{r,s}D_{p,q}\bigr) |f\rangle\langle f|
       \bigl(D_{r,s}D_{p,q}\bigr)^\dagger
       \notag\\
    &= \sum_{p,q\in\mathbb{Z}_d}
       D_{p+r,\,q+s}|f\rangle\langle f|D_{p+r,\,q+s}^\dagger\notag\\
    &= S,\notag
\end{align}
where the last equality holds because each summand depends only on its indices modulo $d$, and the map $(p,q)\mapsto(p+r,q+s)$, with addition taken modulo $d$, is a bijection on $\mathbb{Z}_d\times\mathbb{Z}_d$. As $(p,q)$ ranges over all $d^2$ pairs, so does $(p{+}r,\,q{+}s)$, merely in a different order. Renaming the dummy summation variable $(p{+}r,\,q{+}s)\to(p,q)$ therefore
recovers the original sum~$S$. Hence $S$ commutes with every $D_{r,s}$.

\textbf{Step 2} ($S$ is a multiple of the identity). Since $D_{0,1}=Z$ has the distinct eigenvalues $\omega^j$ and $S$ commutes with $Z$, we have $(\omega^k-\omega^j)\langle j|S|k\rangle=0$, and $\omega^j\neq\omega^k$ for $j\neq k$, so $S$ is diagonal, $S=\sum_j s_j|j\rangle\langle j|$. Since $D_{1,0}=X$, we also have $XSX^\dagger=\sum_j s_j|j\oplus 1\rangle\langle j\oplus 1|=S$, so $s_{j\oplus 1}=s_j$ for all $j$. Hence all $s_j$ are equal and $S=\lambda I_d$ for some $\lambda\in\mathbb{C}$.
 
\emph{Step 3 (trace normalisation).}
Taking the trace of \revise{Eq.}~\eqref{eqn:sic_frame_op} gives
\begin{align*}
    \operatorname{tr}(S)
    &= \sum_{p,q}\operatorname{tr}\!\bigl(D_{p,q}|f\rangle\langle f|D_{p,q}^\dagger\bigr)
    \\
    &= \sum_{p,q}\operatorname{tr}(|f\rangle\langle f|)\\
    &=d^2.
\end{align*}
On the other hand, $\operatorname{tr}(S)=\operatorname{tr}(\lambda I_d)=\lambda d$, we obtain
$\lambda=d$. 
 
\medskip
\noindent\textbf{(iii) Equiangularity at $c_{d^2,d}=1/\!\sqrt{d+1}$.}
For $(p,q)\neq(p',q')\in \mathbb{Z}_d\times \mathbb{Z}_d$, \revise{Eq.}~\eqref{eqn:weyl} gives
\begin{align*}
    \langle\psi_{p',q'}|\psi_{p,q}\rangle
    &= \langle f|D_{p',q'}^\dagger D_{p,q}|f\rangle\\
    &= e^{i\phi}\,\langle f|D_{p-p',\,q-q'}|f\rangle,
\end{align*}
for some phase $e^{i\phi}$.
Since $(p-p',q-q')\not\equiv(0,0)$ modulo $d$, the fiducial
condition \revise{Eq.}~\eqref{eqn:fiducial} of Definition~\ref{def:sic} gives
\begin{align*}
    |\langle\psi_{p',q'} |\psi_{p,q}\rangle|^2
    &= |\langle f|D_{p-p',\,q-q'}|f\rangle|^2\\
    &= \frac{1}{d+1}.
\end{align*}
Taking square roots, $|\langle\psi_{p',q'}|\psi_{p,q}\rangle|=1/\!\sqrt{d+1}$.
 
\medskip
Properties~(i)--(iii) together establish that
$\{|\psi_{p,q}\rangle\}$ is an ETF$(d^2,d)$. The rest follows from the application of Lemma~\ref{thm:tight}.
\end{proof}
 
\begin{remark}[Existence of Fiducial State]
\label{rem:sic_existence}
Proposition~\ref{prop:hw_sic_etf} proves that \emph{if} a fiducial state
exists then the resulting states form an ETF$(d^2,d)$.
It does not establish the existence of fiducial states, which is the content of the
Zauner conjecture~\cite{zauner2011quantum} asserting that fiducial states exist in every dimension $d\geq 1$. Fiducial states have been found in several dimensions~\cite{scott2010symmetric, scott2017sics, scott2010symmetric}.
\end{remark}

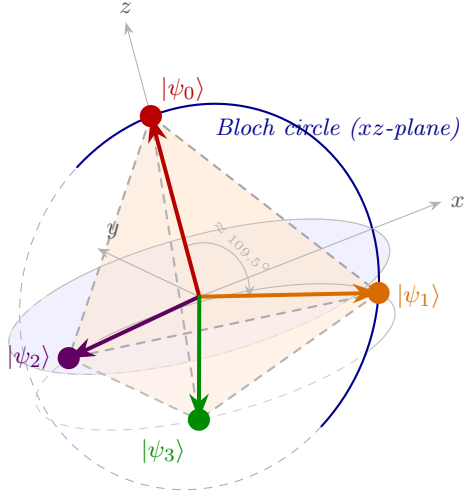
\begin{figure}
    \centering
    \tdplotsetmaincoords{72}{-20}
\begin{tikzpicture}[
  tdplot_main_coords,
  scale=2.6,
  rotate=15,
  >=Stealth,
  font=\small,
  line cap=round,
  line join=round,
]
  \pgfmathsetmacro{\sqtc}{0.81650}
  \pgfmathsetmacro{\sqth}{0.40825}
  \pgfmathsetmacro{\sqhf}{0.70711}
  \pgfmathsetmacro{\rring}{0.94281}
  \colorlet{csphere}{blue!6!white}
  \colorlet{ccirc}  {blue!55!black}
  \colorlet{ceq}    {gray!50}
  \colorlet{cax}    {gray!60}
  \colorlet{cpsi0}  {red!72!black}
  \colorlet{cpsi1}  {orange!85!black}
  \colorlet{cpsi2}  {violet!75!black}
  \colorlet{cpsi3}  {green!55!black}
  \draw[gray!30, fill=csphere] (0,0,0) circle (1);
  \draw[ceq, dashed, thin]
    (-0.42262,0.90631,0) -- (-0.47460,0.88020,0) -- (-0.52498,0.85112,0) -- (-0.57358,0.81915,0) -- (-0.62024,0.78442,0) -- (-0.66480,0.74703,0) -- (-0.70711,0.70711,0) -- (-0.74703,0.66480,0) -- (-0.78442,0.62024,0) -- (-0.81915,0.57358,0) -- (-0.85112,0.52498,0) -- (-0.88020,0.47460,0) -- (-0.90631,0.42262,0) -- (-0.92935,0.36921,0) -- (-0.94924,0.31454,0) -- (-0.96593,0.25882,0) -- (-0.97934,0.20222,0) -- (-0.98944,0.14493,0) -- (-0.99619,0.08716,0) -- (-0.99958,0.02908,0) -- (-0.99958,-0.02908,0) -- (-0.99619,-0.08716,0) -- (-0.98944,-0.14493,0) -- (-0.97934,-0.20222,0) -- (-0.96593,-0.25882,0) -- (-0.94924,-0.31454,0) -- (-0.92935,-0.36921,0) -- (-0.90631,-0.42262,0) -- (-0.88020,-0.47460,0) -- (-0.85112,-0.52498,0) -- (-0.81915,-0.57358,0) -- (-0.78442,-0.62024,0) -- (-0.74703,-0.66480,0) -- (-0.70711,-0.70711,0) -- (-0.66480,-0.74703,0) -- (-0.62024,-0.78442,0) -- (-0.57358,-0.81915,0) -- (-0.52498,-0.85112,0) -- (-0.47460,-0.88020,0) -- (-0.42262,-0.90631,0) -- (-0.36921,-0.92935,0) -- (-0.31454,-0.94924,0) -- (-0.25882,-0.96593,0) -- (-0.20222,-0.97934,0) -- (-0.14493,-0.98944,0) -- (-0.08716,-0.99619,0) -- (-0.02908,-0.99958,0) -- (0.02908,-0.99958,0) -- (0.08716,-0.99619,0) -- (0.14493,-0.98944,0) -- (0.20222,-0.97934,0) -- (0.25882,-0.96593,0) -- (0.31454,-0.94924,0) -- (0.36921,-0.92935,0) -- (0.42262,-0.90631,0);
  \draw[ccirc!35, dashed, thin]
    (0.45710,0,-0.88942) -- (0.41291,0,-0.91077) -- (0.36772,0,-0.92994) -- (0.32165,0,-0.94686) -- (0.27480,0,-0.96150) -- (0.22729,0,-0.97383) -- (0.17923,0,-0.98381) -- (0.13074,0,-0.99142) -- (0.08194,0,-0.99664) -- (0.03294,0,-0.99946) -- (-0.01614,0,-0.99987) -- (-0.06519,0,-0.99787) -- (-0.11407,0,-0.99347) -- (-0.16268,0,-0.98668) -- (-0.21090,0,-0.97751) -- (-0.25861,0,-0.96598) -- (-0.30570,0,-0.95213) -- (-0.35205,0,-0.93598) -- (-0.39755,0,-0.91758) -- (-0.44209,0,-0.89697) -- (-0.48557,0,-0.87420) -- (-0.52788,0,-0.84932) -- (-0.56892,0,-0.82239) -- (-0.60859,0,-0.79349) -- (-0.64679,0,-0.76267) -- (-0.68343,0,-0.73001) -- (-0.71843,0,-0.69560) -- (-0.75170,0,-0.65951) -- (-0.78315,0,-0.62183) -- (-0.81272,0,-0.58266) -- (-0.84033,0,-0.54207) -- (-0.86592,0,-0.50019) -- (-0.88942,0,-0.45710) -- (-0.91077,0,-0.41291) -- (-0.92994,0,-0.36772) -- (-0.94686,0,-0.32165) -- (-0.96150,0,-0.27480) -- (-0.97383,0,-0.22729) -- (-0.98381,0,-0.17923) -- (-0.99142,0,-0.13074) -- (-0.99664,0,-0.08194) -- (-0.99946,0,-0.03294) -- (-0.99987,0,0.01614) -- (-0.99787,0,0.06519) -- (-0.99347,0,0.11407) -- (-0.98668,0,0.16268) -- (-0.97751,0,0.21090) -- (-0.96598,0,0.25861) -- (-0.95213,0,0.30570) -- (-0.93598,0,0.35205) -- (-0.91758,0,0.39755) -- (-0.89697,0,0.44209) -- (-0.87420,0,0.48557) -- (-0.84932,0,0.52788) -- (-0.82239,0,0.56892) -- (-0.79349,0,0.60859) -- (-0.76267,0,0.64679) -- (-0.73001,0,0.68343) -- (-0.69560,0,0.71843) -- (-0.65951,0,0.75170) -- (-0.62183,0,0.78315) -- (-0.58266,0,0.81272) -- (-0.54207,0,0.84033) -- (-0.50019,0,0.86592) -- (-0.45710,0,0.88942);
  \draw[gray!40, dashed, thin]
    (-0.25207,0.90849,-0.33333) -- (-0.30508,0.89209,-0.33333) -- (-0.35703,0.87259,-0.33333) -- (-0.40775,0.85008,-0.33333) -- (-0.45705,0.82462,-0.33333) -- (-0.50477,0.79630,-0.33333) -- (-0.55074,0.76523,-0.33333) -- (-0.59480,0.73150,-0.33333) -- (-0.63681,0.69525,-0.33333) -- (-0.67660,0.65658,-0.33333) -- (-0.71406,0.61564,-0.33333) -- (-0.74904,0.57256,-0.33333) -- (-0.78142,0.52751,-0.33333) -- (-0.81110,0.48062,-0.33333) -- (-0.83797,0.43208,-0.33333) -- (-0.86194,0.38203,-0.33333) -- (-0.88292,0.33066,-0.33333) -- (-0.90085,0.27815,-0.33333) -- (-0.91565,0.22467,-0.33333) -- (-0.92728,0.17041,-0.33333) -- (-0.93570,0.11557,-0.33333) -- (-0.94088,0.06032,-0.33333) -- (-0.94280,0.00487,-0.33333) -- (-0.94145,-0.05060,-0.33333) -- (-0.93684,-0.10590,-0.33333) -- (-0.92899,-0.16083,-0.33333) -- (-0.91792,-0.21520,-0.33333) -- (-0.90367,-0.26883,-0.33333) -- (-0.88629,-0.32153,-0.33333) -- (-0.86584,-0.37311,-0.33333) -- (-0.84239,-0.42340,-0.33333) -- (-0.81602,-0.47222,-0.33333) -- (-0.78683,-0.51941,-0.33333) -- (-0.75491,-0.56480,-0.33333) -- (-0.72038,-0.60823,-0.33333) -- (-0.68335,-0.64956,-0.33333) -- (-0.64395,-0.68863,-0.33333) -- (-0.60233,-0.72532,-0.33333) -- (-0.55861,-0.75950,-0.33333) -- (-0.51296,-0.79105,-0.33333) -- (-0.46554,-0.81985,-0.33333) -- (-0.41650,-0.84582,-0.33333) -- (-0.36602,-0.86886,-0.33333) -- (-0.31427,-0.88889,-0.33333) -- (-0.26144,-0.90584,-0.33333) -- (-0.20770,-0.91965,-0.33333) -- (-0.15323,-0.93027,-0.33333) -- (-0.09824,-0.93768,-0.33333) -- (-0.04291,-0.94183,-0.33333) -- (0.01257,-0.94273,-0.33333) -- (0.06801,-0.94035,-0.33333) -- (0.12321,-0.93472,-0.33333) -- (0.17799,-0.92586,-0.33333) -- (0.23214,-0.91378,-0.33333) -- (0.28550,-0.89854,-0.33333) -- (0.33786,-0.88019,-0.33333) -- (0.38906,-0.85879,-0.33333) -- (0.43891,-0.83442,-0.33333) -- (0.48724,-0.80715,-0.33333) -- (0.53387,-0.77709,-0.33333);
  \fill[orange!10, opacity=0.35]
    (0,0,1) -- (0.9428,0,-0.33333) -- (-0.4714,-0.8164,-0.33333) -- cycle;
  \fill[violet!8, opacity=0.30]
    (0,0,1) -- (-0.4714,0.8164,-0.33333) -- (-0.4714,-0.8164,-0.33333) -- cycle;
  \fill[orange!15, opacity=0.45]
    (0.9428,0,-0.33333) -- (-0.4714,0.8164,-0.33333) --
    (-0.4714,-0.8164,-0.33333) -- cycle;
  \fill[orange!18, opacity=0.55]
    (0,0,1) -- (0.9428,0,-0.33333) -- (-0.4714,0.8164,-0.33333) -- cycle;
  \draw[gray!50, thick, dashed]
    (-0.4714,0.8164,-0.33333) -- (-0.4714,-0.8164,-0.33333);
  \draw[gray!60, thick, dashed]
    (0,0,1) -- (-0.4714,-0.8164,-0.33333);
  \draw[gray!60, thick, dashed]
    (0.9428,0,-0.33333) -- (-0.4714,-0.8164,-0.33333);
  \draw[gray!65, thick, dashed]
    (0,0,1) -- (-0.4714,0.8164,-0.33333);
  \draw[gray!65, thick, dashed]
    (0.9428,0,-0.33333) -- (-0.4714,0.8164,-0.33333);
  \draw[gray!70, thick, dashed]
    (0,0,1) -- (0.9428,0,-0.33333);
  \draw[ceq, thin]
    (0.42262,-0.90631,0) -- (0.47460,-0.88020,0) -- (0.52498,-0.85112,0) -- (0.57358,-0.81915,0) -- (0.62024,-0.78442,0) -- (0.66480,-0.74703,0) -- (0.70711,-0.70711,0) -- (0.74703,-0.66480,0) -- (0.78442,-0.62024,0) -- (0.81915,-0.57358,0) -- (0.85112,-0.52498,0) -- (0.88020,-0.47460,0) -- (0.90631,-0.42262,0) -- (0.92935,-0.36921,0) -- (0.94924,-0.31454,0) -- (0.96593,-0.25882,0) -- (0.97934,-0.20222,0) -- (0.98944,-0.14493,0) -- (0.99619,-0.08716,0) -- (0.99958,-0.02908,0) -- (0.99958,0.02908,0) -- (0.99619,0.08716,0) -- (0.98944,0.14493,0) -- (0.97934,0.20222,0) -- (0.96593,0.25882,0) -- (0.94924,0.31454,0) -- (0.92935,0.36921,0) -- (0.90631,0.42262,0) -- (0.88020,0.47460,0) -- (0.85112,0.52498,0) -- (0.81915,0.57358,0) -- (0.78442,0.62024,0) -- (0.74703,0.66480,0) -- (0.70711,0.70711,0) -- (0.66480,0.74703,0) -- (0.62024,0.78442,0) -- (0.57358,0.81915,0) -- (0.52498,0.85112,0) -- (0.47460,0.88020,0) -- (0.42262,0.90631,0) -- (0.36921,0.92935,0) -- (0.31454,0.94924,0) -- (0.25882,0.96593,0) -- (0.20222,0.97934,0) -- (0.14493,0.98944,0) -- (0.08716,0.99619,0) -- (0.02908,0.99958,0) -- (-0.02908,0.99958,0) -- (-0.08716,0.99619,0) -- (-0.14493,0.98944,0) -- (-0.20222,0.97934,0) -- (-0.25882,0.96593,0) -- (-0.31454,0.94924,0) -- (-0.36921,0.92935,0) -- (-0.42262,0.90631,0);
  \draw[ccirc, thick]
    (-0.45710,0,0.88942) -- (-0.41291,0,0.91077) -- (-0.36772,0,0.92994) -- (-0.32165,0,0.94686) -- (-0.27480,0,0.96150) -- (-0.22729,0,0.97383) -- (-0.17923,0,0.98381) -- (-0.13074,0,0.99142) -- (-0.08194,0,0.99664) -- (-0.03294,0,0.99946) -- (0.01614,0,0.99987) -- (0.06519,0,0.99787) -- (0.11407,0,0.99347) -- (0.16268,0,0.98668) -- (0.21090,0,0.97751) -- (0.25861,0,0.96598) -- (0.30570,0,0.95213) -- (0.35205,0,0.93598) -- (0.39755,0,0.91758) -- (0.44209,0,0.89697) -- (0.48557,0,0.87420) -- (0.52788,0,0.84932) -- (0.56892,0,0.82239) -- (0.60859,0,0.79349) -- (0.64679,0,0.76267) -- (0.68343,0,0.73001) -- (0.71843,0,0.69560) -- (0.75170,0,0.65951) -- (0.78315,0,0.62183) -- (0.81272,0,0.58266) -- (0.84033,0,0.54207) -- (0.86592,0,0.50019) -- (0.88942,0,0.45710) -- (0.91077,0,0.41291) -- (0.92994,0,0.36772) -- (0.94686,0,0.32165) -- (0.96150,0,0.27480) -- (0.97383,0,0.22729) -- (0.98381,0,0.17923) -- (0.99142,0,0.13074) -- (0.99664,0,0.08194) -- (0.99946,0,0.03294) -- (0.99987,0,-0.01614) -- (0.99787,0,-0.06519) -- (0.99347,0,-0.11407) -- (0.98668,0,-0.16268) -- (0.97751,0,-0.21090) -- (0.96598,0,-0.25861) -- (0.95213,0,-0.30570) -- (0.93598,0,-0.35205) -- (0.91758,0,-0.39755) -- (0.89697,0,-0.44209) -- (0.87420,0,-0.48557) -- (0.84932,0,-0.52788) -- (0.82239,0,-0.56892) -- (0.79349,0,-0.60859) -- (0.76267,0,-0.64679) -- (0.73001,0,-0.68343) -- (0.69560,0,-0.71843) -- (0.65951,0,-0.75170) -- (0.62183,0,-0.78315) -- (0.58266,0,-0.81272) -- (0.54207,0,-0.84033) -- (0.50019,0,-0.86592) -- (0.45710,0,-0.88942);
  \draw[gray!55, thin]
    (0.53387,-0.77709,-0.33333) -- (0.58844,-0.73663,-0.33333) -- (0.63994,-0.69236,-0.33333) -- (0.68813,-0.64449,-0.33333) -- (0.73275,-0.59327,-0.33333) -- (0.77356,-0.53898,-0.33333) -- (0.81035,-0.48189,-0.33333) -- (0.84295,-0.42229,-0.33333) -- (0.87116,-0.36051,-0.33333) -- (0.89486,-0.29685,-0.33333) -- (0.91391,-0.23166,-0.33333) -- (0.92821,-0.16526,-0.33333) -- (0.93770,-0.09800,-0.33333) -- (0.94232,-0.03023,-0.33333) -- (0.94206,0.03769,-0.33333) -- (0.93690,0.10541,-0.33333) -- (0.92688,0.17259,-0.33333) -- (0.91205,0.23888,-0.33333) -- (0.89248,0.30392,-0.33333) -- (0.86828,0.36739,-0.33333) -- (0.83958,0.42895,-0.33333) -- (0.80652,0.48828,-0.33333) -- (0.76927,0.54508,-0.33333) -- (0.72803,0.59905,-0.33333) -- (0.68301,0.64991,-0.33333) -- (0.63445,0.69740,-0.33333) -- (0.58259,0.74127,-0.33333) -- (0.52771,0.78129,-0.33333) -- (0.47009,0.81725,-0.33333) -- (0.41003,0.84898,-0.33333) -- (0.34784,0.87629,-0.33333) -- (0.28385,0.89906,-0.33333) -- (0.21839,0.91717,-0.33333) -- (0.15179,0.93051,-0.33333) -- (0.08440,0.93902,-0.33333) -- (0.01657,0.94266,-0.33333) -- (-0.05134,0.94141,-0.33333) -- (-0.11898,0.93527,-0.33333) -- (-0.18601,0.92428,-0.33333) -- (-0.25207,0.90849,-0.33333);
 
  \node[ccirc, font=\scriptsize\itshape, anchor=south]
    at (0.9397, 0.03, 0.4520)
    {Bloch circle ($xz$-plane)};
   \draw[cax, ->] (0,0,0) -- (1.40,0,0) node[right, black!65]{$x$};
  \draw[cax, ->] (0,0,0) -- (0,1.28,0) node[above right, black!65]{$y$};
  \draw[cax, ->] (0,0,0) -- (0,0,1.52) node[above, black!65]{$z$};
  \draw[gray!55, ->, thin]
    (0.00000,0.00000,0.28000) -- (0.04959,0.00000,0.27557) -- (0.09761,0.00000,0.26244) -- (0.14254,0.00000,0.24100) -- (0.18297,0.00000,0.21195) -- (0.21761,0.00000,0.17620) -- (0.24538,0.00000,0.13487) -- (0.26538,0.00000,0.08929) -- (0.27700,0.00000,0.04088) -- (0.27986,0.00000,-0.00883) -- (0.27387,0.00000,-0.05825) -- (0.25923,0.00000,-0.10583);
  \node[gray!60, font=\tiny,rotate=-45] at (0.2951,0.0000,0.1655){$\approx 109.5^\circ$};
  \draw[->, cpsi0, line width=1.5pt] (0,0,0) -- (0,0,1);
  \fill[cpsi0] (0,0,1) circle (1.6pt);
  \node[cpsi0, anchor=west, font=\small, xshift=3pt] at (0,0,1.15){$\ket{\psi_0}$};
  \draw[->, cpsi1, line width=1.5pt] (0,0,0) -- (0.9428,0,-0.33333);
  \fill[cpsi1] (0.9428,0,-0.33333) circle (1.6pt);
  \node[cpsi1, anchor=west, font=\small, xshift=3pt]
    at (0.9428,0,-0.35){$\ket{\psi_1}$};
  \draw[->, cpsi2, line width=1.5pt] (0,0,0) -- (-0.4714,0.8164,-0.33333);
  \fill[cpsi2] (-0.4714,0.8164,-0.33333) circle (1.6pt);
  \node[cpsi2, anchor=east, font=\small, xshift=-3pt]
    at (-0.4714,0.8164,-0.33){$\ket{\psi_2}$};
  \draw[->, cpsi3, line width=1.5pt] (0,0,0) -- (-0.4714,-0.8164,-0.33333);
  \fill[cpsi3] (-0.4714,-0.8164,-0.33333) circle (1.6pt);
  \node[cpsi3, anchor=north east, font=\small, xshift=-2pt, yshift=-2pt]
    at (-0.4714,-0.8164,-0.36){$\ket{\psi_3}$};
 \end{tikzpicture}
    \caption{Qubit SIC-POVM in Remark~\ref{rem:sic_d2} for $d=2$.}
    \label{fig:sic_d2}
\end{figure}

\begin{remark}[Qubit SIC-POVM]
\label{rem:sic_d2}
For $d = 2$, the displacement operators are 
$\{D_{0,0}, D_{1,0}, D_{0,1}, D_{1,1}\} = \{I, X, Z, -Y\}$. The fiducial state
\begin{equation}
|f\rangle = \frac{1}{\sqrt{6}}\left[\sqrt{3+\sqrt{3}}\,|0\rangle 
+ e^{i\pi/4}\sqrt{3-\sqrt{3}}\,|1\rangle\right]
\label{eq:qubit-fiducial}
\end{equation}
satisfies the fiducial condition \revise{Eq.}~\eqref{eqn:fiducial}, and the Heisenberg--Weyl 
orbit $\{|f\rangle, X|f\rangle, Z|f\rangle, -Y|f\rangle\}$ forms an 
$\mathrm{ETF}(4, 2)$. 
There exists a unitary 
$U$ on $\mathbb{C}^2$ that rotates these states so that one sits at 
$\ket{0}$ to get 
\begin{align}
|\psi_0\rangle &= |0\rangle, \notag \\ 
|\psi_k\rangle &= \frac{1}{\sqrt{3}}|0\rangle 
+ \sqrt{\tfrac{2}{3}}\,e^{i 2\pi(k-1)/3}|1\rangle, \quad k = 1, 2, 3,
\label{eq:qubit-sic-states}
\end{align}
which is depicted in Figure~\ref{fig:sic_d2}. Both forms are SIC-POVMs with 
coherence $c_{4,2} = 1/\sqrt{3}$ 
achieving 
$P_{\mathrm{guess}}(\{1/N, \rho^x\}_{x \in \mathcal{X}}) = 1/2$ and $Q(X \to A)_\rho = \log (2)$.
\end{remark}

The significance of SIC-POVMs as optimal encodings is two-fold. First, they saturate the Welch bound with the highest possible coherence $c_{d^2,d}=1/\sqrt{d+1}$.
Second, the self-referential measurement $M_{p,q}^*=(1/d)\ket{\psi_{p,q}}\bra{\psi_{p,q}}$ is simultaneously optimal for discrimination and has the property that the POVM elements are proportional to the codewords.

{
\begin{remark}[Mutually unbiased bases]
Two sets of orthonormal bases $ \mathcal{B}^{k}=\{|\psi_{i}^{k}\rangle: i=1, \dots, d\} $ and $ \mathcal{B}^{\ell}=\{|\psi_{j}^{\ell}\rangle: j=1, \dots, d\} $ are called mutually unbiased if and only if~\cite{PhysRevLett.105.030406}
\begin{equation}
\label{mub}
|\langle\psi_{i}^{k} | \psi_{j}^{\ell}\rangle|^{2}=\left\{\begin{array}{ll}
1 / d & \text { for } k \neq \ell, \\
\delta_{i, j} & \text { for } k=\ell.
\end{array}\right.
\end{equation}
In particular, one can find a maximum of $d+1$ sets of mutually unbiased bases in Hilbert spaces of prime-power dimension $d=p^{k}$, with $p$ being a prime and $k$ a positive integer~\cite{mubreview}. However, it is still an open problem whether $d+1$ sets of mutually unbiased bases exist for arbitrary dimensions~\cite{mubreview,PRXQuantum.3.010101}. If there are $a$ sets of mutually unbiased bases, for these $N=ad$ states, we can prove that $P_{\mathrm{guess}}(\{1/N, \rho^x\}_{x \in \mathcal{X}}) = 1/a$ and $Q(X \to A)_\rho = \log (d)$. Therefore, these states provide an optimal encoding, despite not forming an ETF across different sets. Furthermore, such states have been identified as optimal for quantum detector tomography~\cite{xiao2021optimal}.
\end{remark}
}

\begin{remark}[Iterative algorithm for optimal encoding]
\label{rem:algorithm}
When no closed-form ETF construction is available for the required parameters $(N,d)$, one can maximize $\mathcal{Q}(X\to A)_\rho$, or equivalent $P_{\rm guess}(\{1/N,\rho^x\}_{x\in\X})$ directly by projected subgradient ascent. The leakage for a fixed encoding is $
2^{\mathcal{Q}(X\!\to\!A)_\rho}=\max_{\{F_y\}_{y\in\Y}}\sum_{y\in\Y}\trace(\rho^{x^*(y)}F_y),$
where $\Y=\{1,\dots,d^2\}$ and $x^*(y)\in\argmax_{x}\trace(\rho^x F_y)$. The subgradient with respect to $\rho^x$ is {$\partial_{\rho^x}2^{\mathcal{Q}}=\sum_{y:x^*(y)=x}F_y^*$}, where $\{F_y^*\}$ is the optimal POVM can be computed by the iterative algorithm of~\cite{Farokhi_PRA}. The projected subgradient ascent step is {$\rho^x \gets \Pi\bigl[\rho^x+\mu\,\partial_{\rho^x}2^{\mathcal{Q}}\bigr]$}, where $\mu>0$ is the step size and $\Pi$ projects to the set of rank-one
density operators. For any Hermitian operator $\sigma=\sum_i\lambda_i|i\rangle\langle i|$,
the projection is $\Pi[\sigma]=|i^*\rangle\langle i^*|$ where $i^*\in\argmax_i|\lambda_i|$. At each step, we can move each codeword in the direction of the optimal POVM element and project back to the pure-state manifold.
This is related to the frame-potential gradient flow in~\cite{benedetto2003finite} and the alternating projection method~\cite{tropp2005designing}. The quantum-information framing gives the subgradient a natural interpretation as the optimal discriminating measurement rather than a purely geometric update. Convergence guarantees, extensions to noisy quantum channels, and efficient hardware implementation of the resulting encodings are left as directions for future work.
\end{remark}

\section{Numerical Experiments}
\label{sec:numerical}

\subsection{Investigating Popular Quantum Encoding Policies}

We survey seven encoding strategies that span the space from well-known quantum computing encodings to information-theoretically optimal constructions.
\begin{enumerate}
    \item \textit{Basis encoding}: For $x = 0,\ldots,N-1$, the encoded state is $\rho^x=\ket{x}\bra{x}$, where $\{\ket{0},\ldots,\ket{d-1}\}$ is the standard computational basis of $\H\cong\mathbb{C}^d$. This code is valid only if $N\leq d$. For $N>d$, codewords must cycle through the basis modulo~$d$, which result in severe information loss. By Proposition~\ref{prop:index}, basis encoding is universally optimal in the complete regime of $N\leq d$. Basis encoding is the standard binary representation used in quantum algorithms
    appearing in Grover's algorithm, quantum phase estimation, and the HHL algorithm. 
    \item \textit{Phase (DFT) encoding}: The classical value $x = 0,\ldots,N-1$ is encoded as $\rho^x=\ket{\psi_x}\bra{\psi_x}$ with $ \ket{\psi_x}=\frac{1}{\sqrt{d}}\sum_{j=0}^{d-1} e^{2\pi i xj/N}\,\ket{j}$. Phase encoding is a {tight frame for every $N\geq d$}, achieving $\mathcal{Q}(X\to A)_\rho=\log (d)$. This means phase encoding is the optimal universal encoding for $N\geq d$. However, phase encoding is generally {not equiangular}, that is, the pairwise overlaps $|\bra{\psi_x}\ket{\psi_{x'}}| = |d^{-1}\sum_j e^{2\pi i(x-x')j/N}|$ depend on $x-x'$ and grow toward $1$ as $N$ increases. Phase encoding is the basis of the quantum phase estimation circuit and the quantum Fourier transform. 
    \item \textit{Amplitude encoding}: Following \cite{10555535110653511068}, amplitude encoding represents a classical vector $\mathbf{v}\in\mathbb{R}^d$ as $\ket{\psi_{\mathbf{v}}} = \|\mathbf{v}\|^{-1}\sum_j v_j\ket{j}$. Applied to a class label $x\in\{0,\ldots,N-1\}$, we extract a binary feature vector with a leading bias bit, $\mathbf{b}(x)
    = (1,\;\mathrm{bit}_{n-1}(x),\;\ldots,\;\mathrm{bit}_0(x)) \in\{0,1\}^{n+1}$, $n = \lceil\log N\rceil$,
    where $\mathrm{bit}_k(x)$ is the $k$-th bit of $x$ (most significant bit first) and the leading $1$ ensures $\mathbf{b}(x)\neq\mathbf{0}$ for all $x$. {The codeword is then $\ket{\psi_x}= \|\mathbf{b}(x)\|^{-1}\sum_j b_{j}(x)\ket{j}$.} The natural dimension is $d = n+1 = \lceil\log N\rceil + 1$, giving an exponential compression, i.e., $N$ labels in $O(\log N)$ dimensions. The frame operator  is generally not proportional to identity and the encoding is not a tight frame. Consequently $\mathcal{Q}(X\to A)_\rho<\log (d)$ in general. However, unlike basis encoding, it uses only $d=O(\log N)$ dimensions, exploiting the exponential compression of amplitude encoding.
    \item \textit{Equatorial encoding}: For $x = 0,\ldots,N-1$, the encoded state is $\rho^x=\ket{\psi_x}\bra{\psi_x}$, where $\ket{\psi_x}= \cos\!(\frac{\pi x}{N})\ket{0}
    + \sin\!(\frac{\pi x}{N})\ket{1}$. \revise{We call this encoding equatorial, following the terminology of equatorial qubits~\cite{fan2002equatorial}, whose Bloch vectors lie on a great circle of the Bloch sphere in a coordinate plane.} For $d=2$ (qubit), the $N$ codewords span $\mathbb{C}^2$ for all $N\geq 2$, and the frame operator evaluates to $S=(N/2)I_2$, confirming a tight frame for every $N$. This encoding therefore achieves $\mathcal{Q}(X\to A)_\rho=\log (2)$ for all $N\geq 2$, $d=2$. For $d>2$ the codewords all lie in the two-dimensional subspace $\operatorname{span}\{\ket{0},\ket{1}\}$ and do not span {$\H\cong\mathbb{C}^d$}. The qubit trine ($N=3$) is a special case of equatorial encoding.
    \item \textit{Angle encoding}: For $x = 0,\ldots,N-1$, the encoded state is $\rho^x=\ket{\psi_x}\bra{\psi_x}$, where $\ket{\psi_x}
    = \cos\!(\frac{2\pi x}{N})\ket{0}
    + \sin\!(\frac{2\pi x}{N})\ket{1}$. For $d=N=2$ the codewords are $\ket{0}$ and $-\ket{0}$, {identical} density operators, so $P_{\rm guess}(\{1/N,\rho^x\})=1/2$ (random guessing) and $\mathcal{Q}(X\to A)_\rho=0$. This is the worst possible case. For $d=2$ and $N\geq 3$ the codewords span $\mathbb{C}^2$ and the frame operator again gives a tight frame for $d=2$, recovering $\mathcal{Q}(X\to A)_\rho=\log (2)$. As with equatorial encoding, the subspace restriction to $\{\ket{0},\ket{1}\}$ makes performance suboptimal for $d>2$. \revise{This is the angle encoding of~\cite{larose2020robust} for a single qubit.}
    \item \textit{Hamiltonian encoding}: Encoding data via Hamiltonian is the standard data-encoding gate in variational quantum machine learning~\cite{PhysRevA103032430}. Let $H=\operatorname{diag}(h_0,\ldots,h_{d-1})$ with $h_k = 2k/(d-1)-1\in[-1,1]$ equally spaced eigenvalues, and let $\ket{+} = d^{-1/2}\mathbf{1}$ be the uniform superposition. The Hamiltonian encoding is $\rho^x=\ket{\psi_x}\bra{\psi_x}$, where $\ket{\psi_x}= e^{-iHt_x}|{+}\rangle
    = \frac{1}{\sqrt{d}}\sum_{k=0}^{d-1} e^{-ih_k \pi x/N}\,|k\rangle$ with $t_x = \frac{\pi x}{N}$. The classical value $x$ is encoded as an evolution \emph{time} under the Hamiltonian $H$. Since $e^{-iHt_x}$ is a diagonal unitary, 
    Hamiltonian encoding is a diagonal-phase encoding like phase encoding, with phases $e^{-ih_k\pi x/N}$ replacing $e^{2\pi i xk/N}$. For $d=2$ with $H=-Z=\operatorname{diag}(-1,1)$, the codewords are $\ket{\psi_x} = \frac{1}{\sqrt{2}}(
    e^{i\pi x/N}\ket{0}+ e^{-i\pi x/N}\ket{1}),$ which form a tight frame for every $N$, giving $\mathcal{Q}(X\to A)_\rho=\log (2)$. For $d>2$ the frame is not tight for any $N$, so $Q(X\to A)_\rho<\log d$. Hamiltonian encoding arises naturally in quantum simulation and quantum sensing, where the signal $x$ may represent a physical parameter (field strength, coupling constant) that drives a Hamiltonian $H(x)=xH_0$ for some fixed $H_0$. After a fixed evolution time the probe state $e^{-iH(x)t}\ket{+}$ encodes $x$ as a phase. 
    \item \textit{Random encoding}: For $x = 0,\ldots,N-1$, the encoded state is $\rho^x=\ket{\psi_x}\bra{\psi_x}$, where $\ket{\psi_x}= v_x/\|v_x\|$ with $v_x = u_x^{(1)} + i\,u_x^{(2)}$ such that $u_x^{(1)},u_x^{(2)}$ are independent real Gaussian vectors with mean zero and unit variance. This produces states distributed according to the \emph{Haar measure} on the unit sphere $S^{2d-1}\subset\mathbb{C}^d$. \revise{Since $\mathbb{E}[|\psi\rangle\langle\psi|]=I/d$ for Haar-random $|\psi\rangle$, the strong law of large numbers gives $S/N\to I/d$ almost surely for fixed $d$ and $N\to\infty$. Therefore, random encodings approach tight frames when $N/d$ is large.}
\end{enumerate}

\begin{figure*}[t]
    \centering
    \includegraphics[width=1\linewidth]{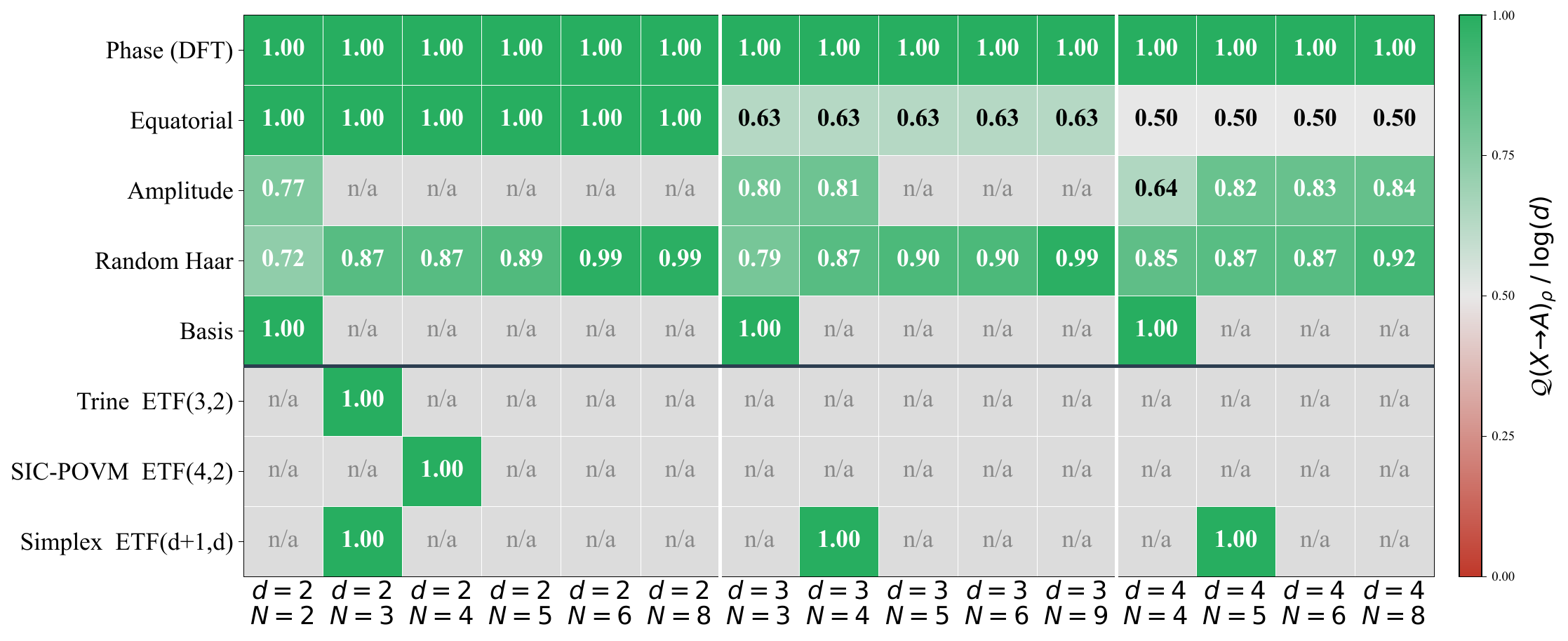}
    \caption{Maximal quantum leakage scaled by dimension, $\mathcal{Q}(X\!\to\!A)_\rho / \log(d)$, across fifteen $(N,d)$ parameter pairs for several encoding policies.}
    \label{fig:heatmap}
\end{figure*}

\begin{figure*}[t]
    \centering
    \definecolor{EncFill}{HTML}{FFFBEB}
\definecolor{EncBord}{HTML}{92400E}
\definecolor{RotFill}{HTML}{EEF2FF}
\definecolor{RotBord}{HTML}{3730A3}
\definecolor{NoiFill}{HTML}{F8FAFC}
\definecolor{NoiBord}{HTML}{94A3B8}
\definecolor{LayFill}{HTML}{F0F9FF}
\definecolor{LayBord}{HTML}{7DD3FC}
\definecolor{WireCol}{HTML}{111827}
\definecolor{ClasBord}{HTML}{374151}
 
\tikzset{
  qg/.style={
    draw=black!58, line width=0.40pt, fill=white,
    rounded corners=3pt, inner sep=2.5pt,
    font=\small, align=center,
  },
  Genc/.style={qg,
    fill=EncFill, draw=EncBord, line width=0.85pt,
    minimum width=1.32cm, minimum height=2.52cm,
  },
  Grot/.style={qg,
    fill=RotFill, draw=RotBord, line width=0.40pt,
    minimum width=1.6cm, minimum height=0.66cm,
  },
  Gnoi/.style={qg,
    fill=NoiFill, draw=NoiBord,
    minimum width=0.75cm, minimum height=0.66cm,
  },
  Gmea/.style={qg,
    draw=black!68, line width=0.40pt, fill=white,
    minimum width=0.86cm, minimum height=0.66cm,
  },
  Gcls/.style={
    draw=ClasBord, double, double distance=0.82pt, line width=0.36pt,
  },
  Gprc/.style={qg,
    draw=ClasBord, line width=0.44pt, fill=gray!4,
    rounded corners=4pt, inner sep=5pt,
  },
}
    \begin{tikzpicture}[line cap=round, line join=round]
    \draw[WireCol, line width=0.72pt]
      (0.78, 0.00) -- (10.92, 0.00);
    \draw[WireCol, line width=0.72pt]
      (0.78,-1.62) -- (10.92,-1.62);
    \node[Genc] (ENC) at (1.42,-0.81)
      {$x\,{\mapsto}$\\[3pt]$|\psi_x\rangle$};
    \node[Gnoi,scale=0.9] (DA0) at (2.96, 0.00) {$\mathcal{D}_p$};
    \node[Gnoi,scale=0.9] (DA1) at (2.96,-1.62) {$\mathcal{D}_p$};
    \node[Grot,scale=0.9] (RY0) at (4.58, 0.00)
      {$R_Y\!(\theta_0^{(\ell)})$};
    \node[Grot,scale=0.9] (RY1) at (4.58,-1.62){$R_Y\!(\theta_1^{(\ell)})$};
    \node[Grot,scale=0.9] (RZ0) at (6.52, 0.00){$R_Z\!(\phi_0^{(\ell)})$};
    \node[Grot,scale=0.9] (RZ1) at (6.52,-1.62){$R_Z\!(\phi_1^{(\ell)})$};
    \draw[WireCol, line width=0.64pt] (7.86, 0.00) -- (7.86,-1.40);
    \filldraw[WireCol] (7.86, 0.00) circle[radius=0.092cm];
    \draw[WireCol, line width=0.64pt] (7.86,-1.62) circle[radius=0.22cm];
    \draw[WireCol, line width=0.64pt] (7.64,-1.62) -- (8.08,-1.62);
    \draw[WireCol, line width=0.64pt] (7.86,-1.40) -- (7.86,-1.84);
    \node[Gnoi] (DB0) at (9.08, 0.00) {$\mathcal{D}_p$};
    \node[Gnoi] (DB1) at (9.08,-1.62) {$\mathcal{D}_p$};
    \node[Gmea] (M0) at (10.48, 0.00) {};
    \coordinate (mb0) at ($(M0.south)!0.30!(M0.north)$);
    \draw[WireCol, line width=0.40pt]
      ($(M0.west |- mb0)+(0.08cm,0)$) -- ($(M0.east |- mb0)+(-0.08cm,0)$);
    \draw[WireCol, line width=0.40pt]
      ($(mb0)+(-0.17cm,0)$)
      arc[start angle=180, end angle=0, radius=0.17cm];
    \draw[-{Stealth[length=2.4pt, width=1.8pt]}, WireCol, line width=0.47pt]
      (mb0) -- ++(34:0.21cm);
     
    \node[Gmea] (M1) at (10.48,-1.62) {};
    \coordinate (mb1) at ($(M1.south)!0.30!(M1.north)$);
    \draw[WireCol, line width=0.40pt]
      ($(M1.west |- mb1)+(0.08cm,0)$) -- ($(M1.east |- mb1)+(-0.08cm,0)$);
    \draw[WireCol, line width=0.40pt]
      ($(mb1)+(-0.17cm,0)$)
      arc[start angle=180, end angle=0, radius=0.17cm];
    \draw[-{Stealth[length=2.4pt, width=1.8pt]}, WireCol, line width=0.47pt]
      (mb1) -- ++(34:0.21cm);
    \node[Gprc, minimum width=1.46cm, minimum height=0.90cm,scale=0.9]
      (AM) at (11.90,-0.81)
      {$\underset{k}{\arg\max}\;p_k$};
    \draw[Gcls, rounded corners=1.5pt] (M0.east) -| (AM.north);
    \draw[Gcls, rounded corners=1.5pt] (M1.east) -| (AM.south);
    \draw[-{Stealth[length=4.8pt, width=3.4pt]}, WireCol, line width=0.64pt]
      (AM.east) -- ++(0.68cm, 0);
    \node[font=\small] at (13.1,-.5)
      {$\hat{Z}$};
    \begin{scope}[on background layer]
      \path[draw=LayBord, line width=0.58pt,
            dashed, dash pattern=on 3.2pt off 1.8pt,
            fill=LayFill, rounded corners=5.5pt]
        (3.60, 0.58) rectangle (9.70,-2.20);
    \end{scope}
    \draw[
      decorate,
      decoration={brace, amplitude=4pt, raise=2.5pt},
      RotBord, line width=0.50pt,
    ]
      (3.60, 0.58) -- (9.7, 0.58)
      node[midway, above=8pt, font=\footnotesize\bfseries, text=RotBord]
        {Trainable unitary $U(\boldsymbol{\theta})$,\quad $\times L \in \{1,2,4,8,16\}$};
    \node[
      above=0.24cm of ENC.north,
      font=\footnotesize\bfseries, text=EncBord, align=center,
    ] {State\\[-1pt]encoding};
    \node[
      font=\footnotesize\bfseries, text=black!52, align=center,
    ] at (10.62, 1.04) {Readout};
    \draw[black!18, line width=0.5pt, dashed,xshift=1.4mm]
      (2.20, 0.78) -- (2.20,-2.00);
    \draw[black!18, line width=0.5pt, dashed, xshift=-0.6mm]
      (9.94, 0.78) -- (9.94,-2.00);
    \end{tikzpicture}
    \caption{Circuit schematic of the variational quantum circuit used for demonstrating the optimal encoding.}
    \label{fig:circuit}
\end{figure*}
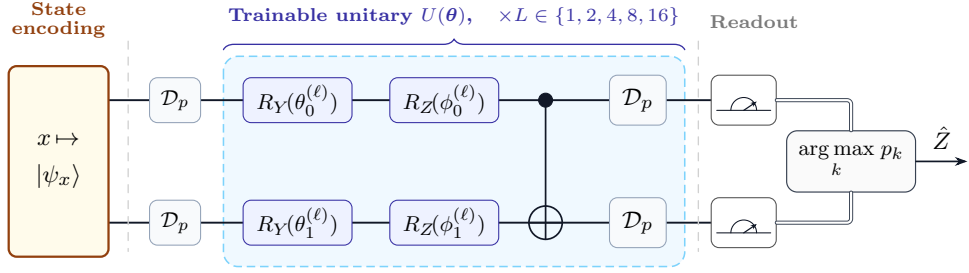

Lemma~\ref{thm:tight} proves that tight frames are optimal encodings as they attain $\mathcal{Q}(X\to A)=\log (d)$ for $N\geq d$. Phase, equatorial (qubit), and Hamiltonian encodings are tight frames for $d=2$ at all $N$. For $d>2$ only Phase and the explicit ETF constructions (trine, SIC-POVM, simplex) maintain this property. 
Equatorial and angle encodings for $d>2$ confine all codewords to the two-dimensional $\{|0\rangle,|1\rangle\}$ subspace. They fail the spanning condition of Lemma~\ref{thm:tight} and achieve $Q(X\to A)_\rho\leq\log(2)<\log(d)$. Amplitude encoding is not a tight frame in general and also achieves $Q(X\to A)_\rho<\log(d)$. Among all tight frames achieving {$\mathcal{Q}(X\to A)_\rho=\log (d)$}, ETFs uniquely satisfy the Welch bound. Phase encoding achieves $\mathcal{Q}(X\to A)_\rho=\log (d)$ but with coherence far exceeding the Welch bound while the trine, SIC-POVM, and simplex achieve both in their corresponding feasible parameter regions. Finally, random encoding is a strong practical solution. For large $N/d$, the frame operator of random states approaches $(N/d)I_d$ and random encodings achieve near-optimal $\mathcal{Q}(X\to A)_\rho$ without any design effort. Explicit ETFs are preferred only when the Welch-bound coherence constraint (e.g.\ for quantum key distribution or tomography) matters.

Figure~\ref{fig:heatmap} reports the fraction of the theoretical maximum leakage achieved by each encoding $\mathcal{Q}(X\!\to\!A)_\rho / \log(d)$ across fifteen $(N,d)$ parameter pairs and five common encoding policies with the three optimal ETF constructions shown below the separator for reference. Several patterns are immediately apparent. Phase (DFT) encoding achieves the maximum $\mathcal{Q}(X\!\to\!A)_\rho / \log(d) = 1.00$
in every cell, confirming that it forms a tight frame for all $N\geq d$ regardless of dimension. Equatorial encoding matches this performance for $d=2$ across all $N$, but degrades as $d$ increases because all codewords lie in the two-dimensional subspace $\mathrm{span}\{\ket{0},\ket{1}\}$ and fail to span $\mathbb{C}^d$ when $d>2$ violating the spanning condition. Amplitude encoding is applicable only when $d \geq \lceil\log N\rceil + 1$, and achieves substantially suboptimal leakage but improving as $N/d$ grows. Random Haar encoding approaches optimality for large $N/d$, \revise{consistent with the law of large numbers argument that Haar-random states approximate tight frames asymptotically}, yet never exactly attains the ceiling.
Basis encoding achieves optimality, however, it is only applicable when $N\leq d$. The three ETF rows confirm the paper's main results on their optimality. 

\subsection{Classification via Quantum Machine Learning}
\label{sec:exp2}
We study a $4$-class quantum classification task.
The input alphabet is $\mathcal{X}=\{0,\ldots,7\}$ ($N=8$), the
Hilbert space is $\mathcal{H}\cong \mathbb{C}^{4}$ ($d=4$, two qubits). The output label $Z\in\{0,1,2,3\}$ is drawn from a \emph{fixed} balanced random partition of $\mathcal{X}$ ({two tokens per class}), chosen once and held constant across all encodings and seeds. With a uniform input prior, $\max_z P\{Z=z\}=\tfrac{1}{4}$.

The variational quantum circuit (VQC) consists of $L\in\{1,2,4,8,16\}$ layers, each comprising per-qubit $(R_Y$--$R_Z)$ rotations followed by a CNOT gate and per-qubit depolarizing noise with rate $p=0.01$ per qubit per layer, matching a
realistic near-term device. The same ansatz is used with every encoding so that differences in
classification accuracy are attributable solely to the encoding itself. State preparation is also subject to the same depolarizing channel. Parameters are trained with the Adam optimizer (learning rate $0.05$) using exact parameter-shift gradients and the cross-entropy loss, evaluated on all $N=8$ input-output pairs per step. Results are averaged over 10 independent random initializations. Figure~\ref{fig:circuit} illustrates the variational quantum circuit. 

We compare four encodings of phase (DFT), random Haar, amplitude, and equatorial encoding. The exact maximal quantum leakage {$Q(X\to A)_\rho=\log(N \cdot P_{\rm guess}(\{1/N,\rho^x\}_{x\in\X}))$} is computed for each
encoding via a semi-definite program, with results reported in Table~\ref{tab:exp2}. As expected, phase encoding achieves the theoretical maximum $Q(X\to A)_\rho=\log (d) = 2$  bits, with random Haar within 5\% of this ceiling, \revise{consistent with the law of large numbers argument}.
 
\begin{table}[h]
\centering
\caption{Encodings used in classification via quantum machine learning ($N=8$, $d=4$).}
\label{tab:exp2}
\begin{tabular}{lcc}
\hline
Encoding & $Q(X \to A)_\rho$ & $P_{\rm guess}(\{1/N,\rho^x\}_{x\in\X})$ \\
\hline
Phase & $2.000$ & $0.5000$ \\
Random Haar & $1.910$ & $0.4697$ \\
Amplitude & $1.678$ & $0.4000$ \\
Equatorial & $1.000$ & $0.2500$ \\
\hline
\end{tabular}
\end{table}

\begin{figure}[t]
  \centering
  \includegraphics[width=0.99\columnwidth]{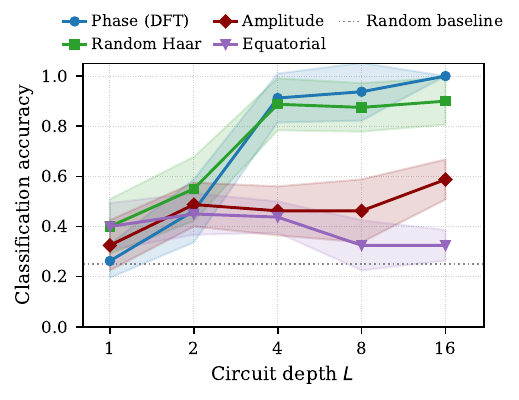}
  \caption{Classification accuracy versus circuit depth $L$ with shaded bands illustrating the standard deviation among various random initializations.}
  \label{fig:acc_vs_depth}
\end{figure}

Figure~\ref{fig:acc_vs_depth} plots classification accuracy against
circuit depth $L$. The results stratify into three tiers that map precisely onto the
maximal quantum leakage value ranking in Table~\ref{tab:exp2}. In Tier~1,
phase encoding sits as an exact tight frame attaining maximal quantum leakage, with random Haar as an approximate tight frame attaining near-maximal leakage. Both encodings reach near-perfect or perfect classification by $L=8$. Phase encoding achieves $100\%$ accuracy at $L=16$ while random Haar reaches $87.5\%$--$100\%$.
In Tier~2, amplitude encoding sits as an example of a partial frame. Accuracy grows from $\approx 31\%$ at $L=1$ to $\approx 56\%$ at $L=16$, yet plateaus well below Tier~1. No amount of additional circuit depth bridges this gap. Binary amplitude encoding achieves the exponential dimensional compression promised by amplitude encoding, but it sacrifices leakage relative to tight-frame alternatives.
At the bottom, in Tier~3, equatorial sits as an example of encoding that wastes the quantum system's potential by only confining the encoding to two dimensions. Equatorial encoding degrades toward the random baseline at large depth, as the circuit overfits a fundamentally two-dimensional measurement structure. Figure~\ref{fig:training_curves} shows epoch-by-epoch accuracy at fixed depth $L=4$. Amplitude converges to an intermediate plateau while phase encoding climbs steeply to near-perfect performance and equatorial collapses toward the random baseline.

\begin{figure}[t]
  \centering
  \includegraphics[width=0.99\columnwidth]{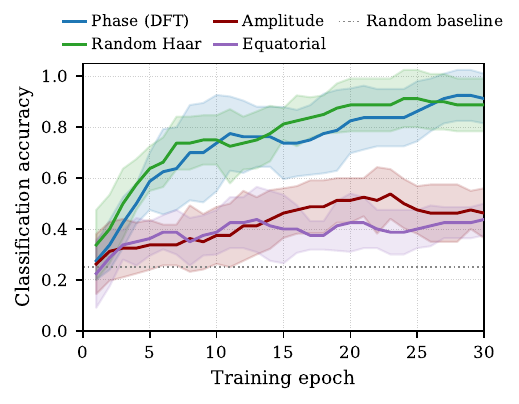}
  \caption{Training curves at fixed depth $L=4$.}
  \label{fig:training_curves}
\end{figure}

\section{Conclusions}
\label{sec:conclusion}
We have developed and validated a complete information-theoretic theory of optimal universal quantum encoding for statistical inference. This is done by rigorously establishing that quantum maximal leakage is the figure of merit for measuring quality of quantum encoding of classical data. Following this, we can compute the optimal encoding by maximizing maximal quantum leakage. The central result is an exact two-regime characterization of the optimal
encoder. In the complete regime ($N \leq d$), basis encoding achieves the
absolute maximum $Q(X\to A)_\rho = \log (N)$. However, in the overcomplete regime ($N > d$), the maximum leakage $Q(X\to A)_\rho = \log (d)$ is achieved  if the codewords form a tight frame with phase encoding being one such encoding. Given tight frames are not unique, we focus on symmetric choices by investigating equiangular tight frames (ETFs) as the
uniquely symmetric optimal encodings.
The numerical experiments corroborate the theory across two complementary
programmes. The leakage heatmap confirms that phase encoding attains largest maximal quantum leakage for every parameter choice when $N\geq d$, while random Haar encoding approaches this ceiling asymptotically. We also use a classification experiment to reveal a clean performance hierarchy aligned with the maximal quantum leakage value ranking, with tight-frame encodings (phase exactly, random Haar asymptotically) performing best.

\revise{Three} directions for future work are particularly natural. First, the present theory assumes ideal state preparation. Extending the optimal characterization to noisy channels would replace the tight-frame
condition with a channel-dependent analogue and is directly relevant to the
noisy intermediate-scale quantum (NISQ) settings studied experimentally. Second, the optimal encodings identified here, e.g., ETFs, SIC-POVMs, but tight frames more broadly, may require state-preparation circuits of substantial depth or non-Clifford gate count, costs that are prohibitive on near-term hardware. A resource-aware theory of quantum encoding would characterize the best trade-off between maximal quantum leakage and the minimum circuit complexity required to realize the codewords, measured in, for example, two-qubit gate count, $T$-gate
count in the fault-tolerant setting, or entanglement cost. Such a framework would transform the present information-theoretic
optimality conditions into practically actionable design principles, making explicit the price in quantum resources that must be paid for each additional bit of leakage and, conversely, identifying the cheapest encoding that meets a prescribed leakage target. \revise{A third direction is to extend the optimal encodings from rank-one to positive operators of any rank, that is, from ETFs and SIC-POVMs to generalized ETFs and generalized equiangular measurements~\cite{siudzinska2022,gour2014}. Although Proposition~\ref{prop:maximal} shows that pure states maximize leakage, mixed codewords arise under noisy preparation. The argument of Lemma~\ref{thm:fund_bound} gives $\mathrm{tr}(\rho^xM_x)\leq\lambda_{\max}(\rho^x)\mathrm{tr}(M_x)$, and hence $\mathcal{Q}(X\to A)_\rho\leq\log\big(d\max_x\lambda_{\max}(\rho^x)\big)$. This is attained by $\rho^x=P_x/r$ with rank-$r$ projectors satisfying $\sum_xP_x=(Nr/d)I$, using $M_x=(d/(Nr))P_x$, which yields $\mathcal{Q}(X\to A)_\rho=\log(d/r)$. Equiangularity would require $\mathrm{tr}(P_xP_{x'})$ to be constant over $x\neq x'$, and taking the trace of the square of $\sum_xP_x$ shows that this constant equals $r(Nr-d)/(d(N-1))$, which reduces to the constant in Eq.~\eqref{eqn:welch} for $r=1$.}

\bibliographystyle{apsrev4-2}
\bibliography{apssamp}

\end{document}